\documentclass[12pt,reqno]{amsart}
\usepackage{amssymb}
\usepackage{url} 
\usepackage{units} 
\usepackage[normalem]{ulem} 
\usepackage{booktabs} 
\usepackage[bookmarks=false,unicode,colorlinks,urlcolor=red,citecolor=red,linkcolor=blue]{hyperref} 
\usepackage[usenames, dvipsnames]{color}
\usepackage{aliascnt}
\usepackage{tikz}
\usepackage[margin=1.25in]{geometry}
\usepackage{float}
\usepackage{pgfplots} 
\pgfplotsset{compat=1.8}
\usepackage{graphicx}

\makeatletter
\def\tagform@#1{\maketag@@@{\ignorespaces#1\unskip\@@italiccorr}}
\let\orgtheequation\theequation
\def\theequation{(\orgtheequation)}
\makeatother
\def\equationautorefname~{}

\newtheorem{theorem}{Theorem}

\newaliascnt{lemma}{theorem}
\newtheorem{lemma}[lemma]{Lemma}
\aliascntresetthe{lemma}

\newaliascnt{proposition}{theorem}
\newtheorem{proposition}[proposition]{Proposition}
\aliascntresetthe{proposition}

\newaliascnt{corollary}{theorem}

\aliascntresetthe{corollary}

\newaliascnt{conjecture}{theorem}

\aliascntresetthe{conjecture}

\newaliascnt{example}{theorem}

\aliascntresetthe{example}

\theoremstyle{remark}
\newtheorem*{remark}{Remark}
\newaliascnt{definition}{theorem}
\newtheorem{definition}[definition]{Definition}
\aliascntresetthe{definition}

\newcommand{\RR}{{\mathbb{R}}}
\newcommand{\NN}{{\mathbb{N}}}

\newcommand{\fodd}{{f_\mathrm{odd}}}
\newcommand{\godd}{{g_\mathrm{odd}}}
\newcommand{\bodd}{{b_\mathrm{odd}}}
\newcommand{\feven}{{f_\mathrm{even}}}
\newcommand{\geven}{{g_\mathrm{even}}}
\newcommand{\beven}{{b_\mathrm{even}}}
\newcommand{\foddinv}{{f_{\mathrm{odd}}^{-1}}}
\newcommand{\feveninv}{{f_{\mathrm{even}}^{-1}}}
\newcommand{\goddinv}{{g_{\mathrm{odd}}^{-1}}}
\newcommand{\geveninv}{{g_{\mathrm{even}}^{-1}}}

\begin{document}

\title[Spectrum of Free rod]{Spectrum of the free rod under tension and compression}
\author{L. Mercredi Chasman}
\address{Division of Science \& Mathematics, University of Minnesota -- Morris, 
MN 56267, U.S.A.}
\email{chasmanm\@@morris.umn.edu}
\author{Jooyeon Chung}
\address{Department of Mathematics, University of Illinois, Urbana,
IL 61801, U.S.A.}
\email{jchung50\@@illinois.edu}
\date{\today}


\begin{abstract}
In this paper, we study the spectrum of the one-dimensional vibrating free rod equation $u^{(4)}-\tau u''=\mu u$ under tension $(\tau>0)$ or compression $(\tau<0)$. The eigenvalues $\mu$ as functions of the tension/compression parameter $\tau$ exhibit three distinct types of behavior. In particular, eigenvalue branches in the lower half-plane exhibit a cascading pattern of barely-avoided crossings. 

We provide a complete description of the eigenfunctions and eigenvalues by implicitly parameterizing the eigenvalue curves. We also establish properties of the eigenvalue curves such as monotonicity, crossings, asymptotic growth, cascading and phantom spectral lines. 
\end{abstract}

\maketitle

\section{\bf Introduction}
This paper investigates the spectrum of a one-dimensional vibrating free rod under tension or compression, which exhibits unexpected behaviors in the compressive regime.  Since the rod is of fixed length, we may take our domain to be $\Omega = (-1, 1)$; the spectrum of rods of other lengths can be recovered by rescaling.

The eigenvalues $\mu$ of the free rod depend on a tension parameter $\tau$ (discussed below) and are governed by the differential equation
\begin{equation}
u^{(4)} - \tau u'' = \mu u \label{eqn:1dimDE}
\end{equation}
together with the boundary conditions
\begin{align} 
u'' &= 0 \qquad \text{at $x = \pm 1$,}  \label{eqn:1dbcm}\\
u''' -\tau u'  &= 0 \qquad\text{at $x = \pm 1$.} \label{eqn:1dbcv}
\end{align} 
These boundary conditions arise naturally from the minimizers of the rod Rayleigh quotient, which takes the form
\begin{equation}
 Q[u] = \frac{\int_{-1}^{1}|u^{\prime\prime}|^2 + \tau |u^\prime|^2\,dx}{\int_{-1}^{1}|u|^2\,dx}. \label{1dimRQ}
\end{equation}
It is straightforward to show (see \cite{Chasman11}) the Rayleigh quotient is coercive, and so the free rod eigenvalue problem has a complete discrete spectrum with an orthonormal eigenbasis.

Interpreted physically, the parameter $\tau$ represents the tension applied to the ends of the rod. The sign of the tension parameter $\tau$ determines whether the rod is under tension ($\tau>0$) or compression ($\tau<0$). Positive eigenvalues correspond to a vibrating rod, while $\mu<0$ indicates the rod is buckling. The case of eigenvalue $\mu=0$ corresponds to a translational mode.

\begin{figure}[t]
    \begin{center}
\includegraphics[scale=0.5]{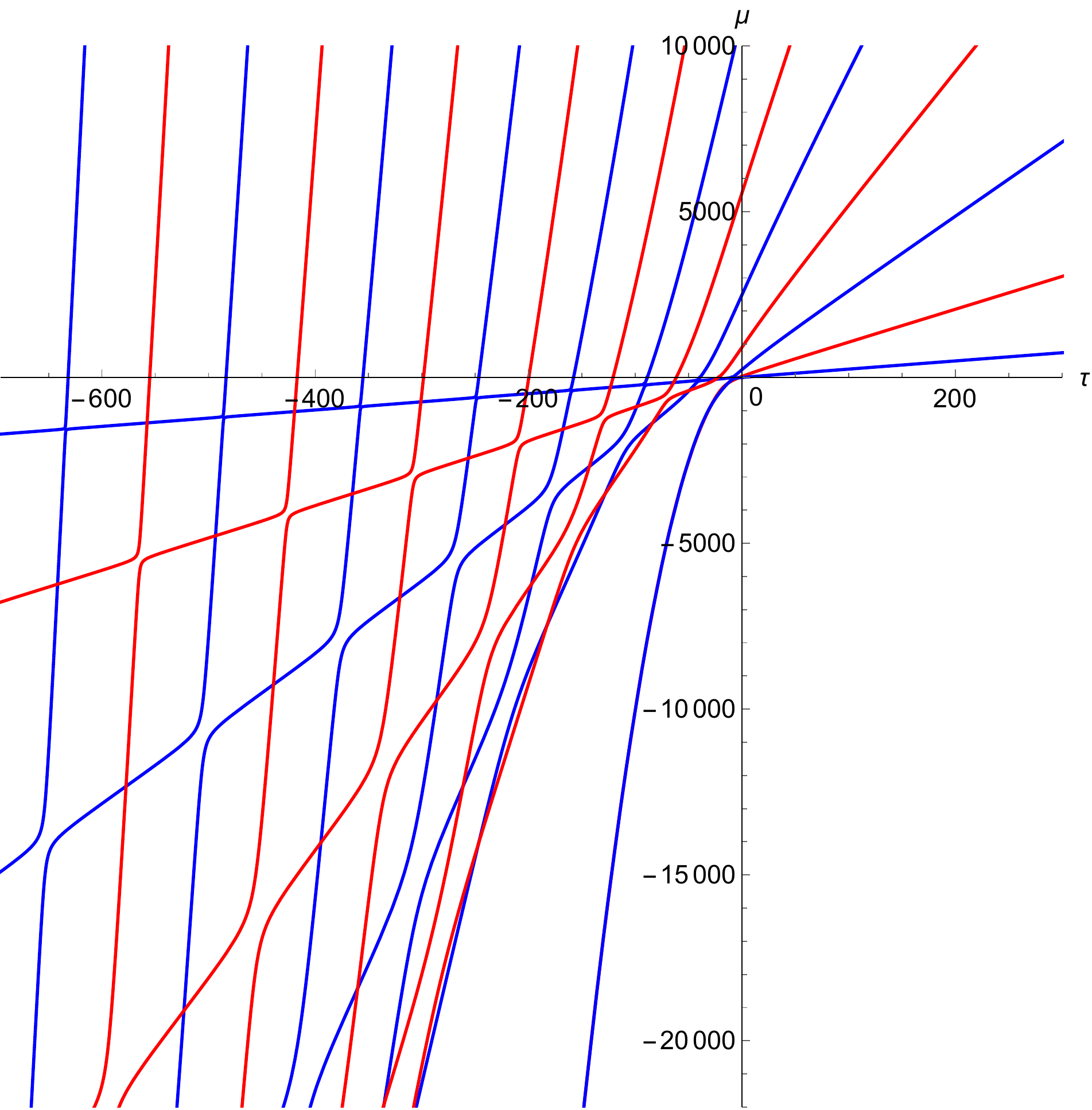}
\caption{\label{fig:fullspec} Spectrum of the free rod showing behavior under tension ($\tau>0$) and compression ($\tau<0$) with vibrational modes ($\mu>0$) and buckling modes ($\mu<0$). The spectrum shows repeated cascading behavior in the lower half-plane. Blue curves are eigenvalues associated with odd eigenfunctions; red are associated with even eigenfunctions.}
\end{center}
\end{figure}

The spectrum exhibits three types of behavior, shown in Figure~\ref{fig:fullspec}. In the upper half-plane, we observe nearly-linear, non-intersecting eigenvalue branches alternating based on symmetry of the associated eigenfunctions. Behavior in the lower half-plane depends on whether the eigenvalue branches lie above or below the parabola $\mu=-\tau^2/4$, shown more clearly in Figure~\ref{fig:zoomspec}. We will refer to this parabola as the \emph{critical parabola}, dividing the lower half-plane into \emph{sub-} and \emph{super-parabolic} regions.  Only two spectral curves penetrate the sub-parabolic region, below the critical parabola; we refer to these as the first two buckling branches. Above the critical parabola, we see a pattern of barely-avoided crossings along eigenvalue branches of the same symmetry (called \emph{cascading}), which we discuss further at the end of this section. We can also see that the spectrum has predictable asymptotic behavior in this region. We are particularly interested in how the behavior of the eigenvalue curves changes as we move from a vibrational mode ($\mu>0$) to a buckling mode ($\mu<0$).

\begin{figure}[t]
    \begin{center}
\includegraphics[scale=0.5]{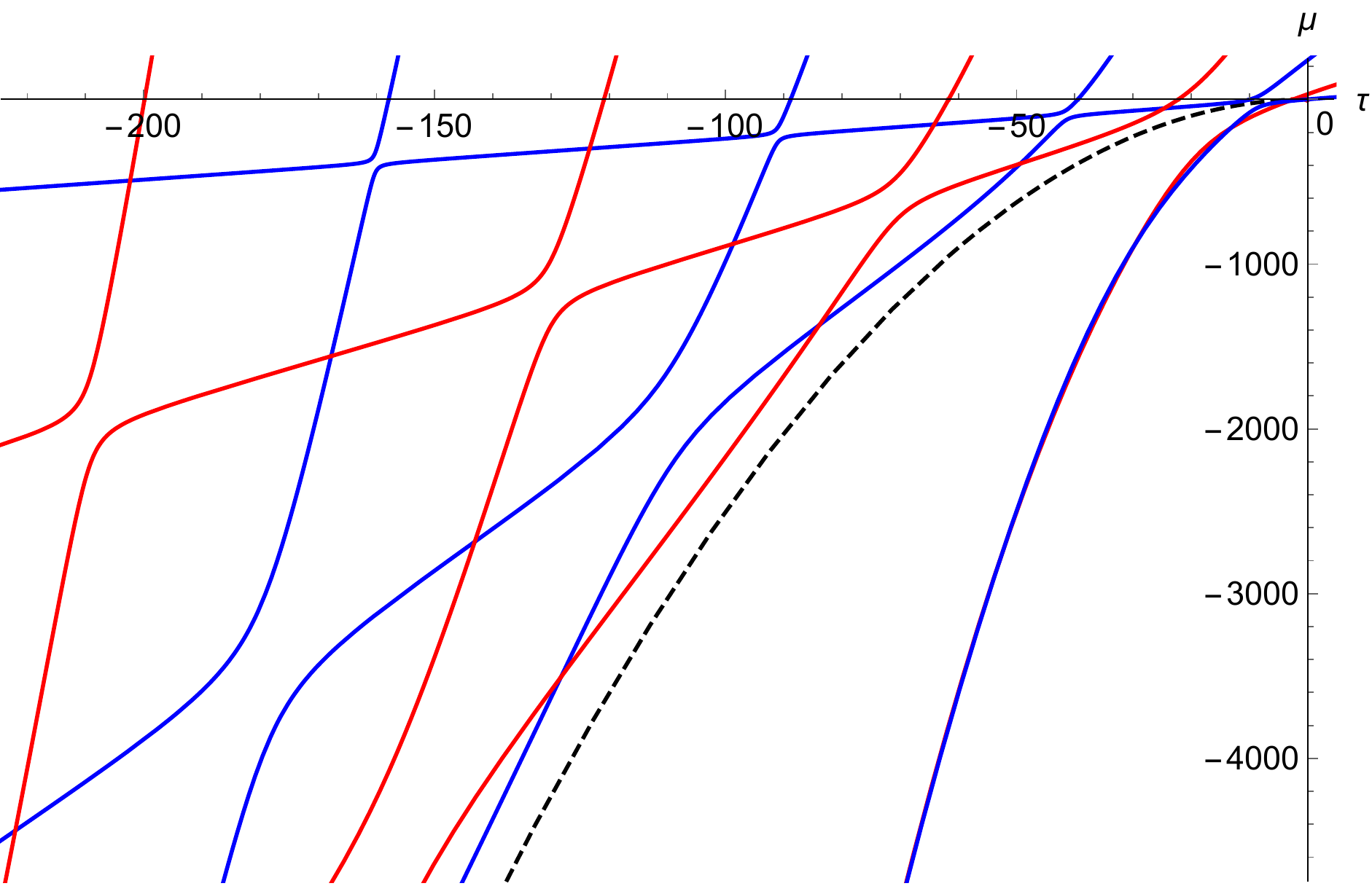}
\caption{\label{fig:zoomspec} Spectral behavior in the lower half-plane. The dashed curve is the critical parabola $\mu=-\tau^2/4$ (not part of the spectrum), which divides two different types of solutions. Cascading occurs above the critical parabola. Below the critical parabola, there are only two eigenvalue curves, which intersect infinitely often (see Figure~\ref{fig:crit} later).}
    \end{center}
    \end{figure}
    
We study the eigenvalues $\mu_k$ as functions of the tension parameter $\tau$ and consider rods under both tension and compression. In order to completely identify the spectrum, we parameterize the eigenvalue curves for each $\mu_k(\tau)$ for all real values of $\tau$. Different parameterizations are used in three different regions in the $(\tau, \mu)$-plane: the upper half-plane, the super-parabolic region (third quadrant above the parabola $\mu=-\tau^2/4$), and the sub-parabolic region (third quadrant below the parabola $\mu=-\tau^2/4$). Note that the fourth quadrant contains no eigenvalues. This is easiest to see from the Rayleigh quotient \eqref{1dimRQ}, whose numerator is nonnegative when $\tau > 0$.

This paper focuses on the interval $(-1,1)$ of length $2$. The general interval case can then be obtained from translation and the follow scaling relation:
\[
\mu_j\big((-R, R), \tau\big) = \frac{1}{R^4}\mu_j\left((-1, 1),  R^2\tau\right), \qquad j=1, 2, 3, \dots
\]

The rod eigenvalue equation can be considered with other boundary conditions. To the best of our knowledge, the spectrum of clamped rod (with $u=u'=0$ at the endpoints) has not been analyzed in the manner of this paper.

The structure of the paper is as follows. In Section~\ref{sec:Prelim}, we establish properties of symmetry of the eigenfunctions, and introduce bijections of regions of the $(\tau, \mu)$-plane that we will use for our parameterizations of eigenvalue branches. In Sections~\ref{sec:UHP} and \ref{sec:LHP}, we analyze the eigenvalues in the upper half-plane and super-parabolic region, respectively, by finding eigenvalue conditions for each region and then parameterizing the eigenvalue branches. We also discuss monotonicity, crossing properties, asymptotic growth of the eigenvalue branches, cascading and phantom lines in the spectrum. In Section~\ref{sec:crit}, we analyze the eigenvalues in the sub-parabolic region. Our approach to this region differs from the other regions, since our usual parameterization approach will not work. We find the eigenvalue conditions and describe the behavior of the two eigenvalue branches that lie in this region. In addition, we establish a result involving intersections of a family of parabolas with the eigenvalue branches. We also identify the crossings of the odd and even eigenvalue branches. 

\subsection*{\textbf{Related literature}}
The free rod is the one-dimensional case of the free plate. Plate problems are fourth-order analogs of membrane problems, with the bi-Laplacian operator taking the place of the Laplacian. The bi-Laplacian is more difficult to work with, as it is fourth-order and lacks some standard properties of the Laplacian. For instance, the maximum principle does not hold for the bi-Laplacian. However, fourth-order problems with appropriate boundary conditions have modeled a number of plates with physically relevant conditions; for instance, Sweers \cite{Sweers(16)} recently gave a survey of sign- and positivity-preserving properties of rod and plate problems with certain boundary conditions.

The literature includes a number of papers on fourth-order eigenvalue problems involving a parameter playing a similar role to our $\tau$. Notable work includes that of Kawohl, Levine, and Velte \cite{KawohlLevineVelte93}, who considered eigenvalues of a clamped vibrating plate under tension $\tau$:
\begin{align*}
\Delta^2 u-\tau \Delta u&=\gamma u
\end{align*} 
with clamped boundary conditions $u=|\nabla u|=0$ on the boundary. Hence $\tau>0$ corresponds to tension, and $\tau<0$ to compression. In fact, one can instead regard $\gamma$ as the parameter and $\tau$ as the eigenvalue, in which case one obtains the buckling problem. They proved concavity with respect to the parameter of sums of low eigenvalues. Their paper treats the higher-dimensional case, where much less is known. More recently, Ashbaugh, Benguria, and Mahadevan \cite{ABM17} proved an isoperimetric inequality for the first eigenvalue of the clamped plate under tension for a small range of $\tau<0$. We expect that in the one-dimensional clamped rod problem, one could obtain an explicit parameterization of the spectrum and study properties such as cascading and eigenvalue crossings, similar to our work for the free rod problem in this paper. 

Spectral problems with a tension-type parameter result naturally in a family of eigenvalue curves. For example, Grunau \cite{Grunau02} considered the related one-dimensional buckling eigenvalue problem $u^{(4)}+ au'''=-\lambda u''$ with clamped boundary conditions. He described the spectral curves as functions of a parameter $a$, and found that in appropriate parameter domains, these curves look different from those for the same equation under Navier boundary conditions.

We do not expect to get an explicit parameterization of eigenvalue curves in the higher-dimensional case, since the spectrum depends on the shape of the plate. However, one can establish relationships between plate and membrane eigenvalues, and between the plate eigenvalues and buckling energies ($\tau$ values for $\mu=0$), in the forms of inequalities. Payne \cite{Payne56} derived such inequalities for both the eigenvalues of the buckling problem and the vibrational eigenvalues for the clamped plate. These include linear (in $\tau$) upper bounds on single vibrational eigenvalues, and linear lower bounds on their sums, with coefficients given by the buckling energies. For the free rod, we observe nearly-linear behavior for all positive (vibrational) eigenvalues (see Section~\ref{sec:UHP}), with approximate slope given by the free membrane eigenvalues. Linear relationships between free rod eigenvalues and free rod buckling energies also appear in the phantom spectral lines.

 
The phenomena of avoided crossings and cascading arise in a variety of spectral problems, and have been studied in a number of second-order problems with a parameter. These terms are not well-defined mathematically (although \cite{Hineman07} provides a precise definition of avoided crossing in the context of their work), and instead are generally identified visually, as qualitative properties of the spectrum. We use the term \emph{avoided crossing} if two spectral curves come close together, nearly intersecting, but then veer away sharply. By \emph{cascading}, we mean that as $\tau$ increases, a spectral curve exhibits drastic changes in slope at nearly-periodic intervals, with a relatively steady rate of increase between these transition periods. These regions of alternating steady-then-sharp-increase create a pattern of \emph{phantom spectral lines}, as discussed in Section ~\ref{sec:propertyLHP}(\ref{prop:LHPphantom}).

Avoided crossings (also called quasi-crossings) of eigenvalues for a Coulomb centers problem were first studied by Komarov and Slavyanov \cite{Komarov68}. To the best of the authors' knowledge, the term cascading was first used by Gesztesy \emph{et al.} in \cite{Gesztesy88}, which investigated cascading of eigenvalues of a family of Schr\"odinger eigenvalue problems. Also, avoided crossings and cascading occur in a family of Heun's DE problems. For instance, Slavyanov and Veshev \cite{Slavyanov97} showed in 1997 that avoided crossings of a particular family occur periodically with respect to the parameter, and Bay \emph{et al.} \cite{Bay97} calculated avoided crossings of eigenvalue curves of the quartic oscillator of Heun's DE and showed dependence of the gap of avoided crossings on asymmetry of the parameter. More recently, in 2007, Hineman and Neuberger \cite{Hineman07} considered avoided crossings of eigenvalues of nonlinear second-order PDE's on certain regions and suggested some numerical techniques to solve these problems. 

Avoided crossings and cascading are in principle distinct phenomena, but they seem to be connected, since cascading occurs when there is a nearly-periodic pattern of avoided crossings, such as in this paper.

\section{\bf Preliminaries}\label{sec:Prelim} 
In this section, we establish two results that will aid our treatment of the spectrum. We prove a result about the symmetry of the eigenfunctions which allows us to assume all eigenfunctions are either odd or even on the open interval $(-1, 1)$, simplifying the solving of the boundary value problem.  We then define three bijections of portions of the $(\tau,\mu)$-plane, which will allow us to parameterize the spectral curves.

\subsection{\bf Reduction to odd and even eigenfunctions}

Before embarking on the classification of the eigenvalues, we show that we need only consider odd and even solutions to the eigenvalue equation \eqref{eqn:1dimDE}. 

Note that if $u(x)$ is an eigenfunction satisfying the partial differential equation \eqref{eqn:1dimDE} and the boundary conditions \eqref{eqn:1dbcm} and \eqref{eqn:1dbcv}, then by symmetry so is $u(-x)$, with the same eigenvalue $\mu$. The odd and even parts of $u$ can be expressed as
\[
u_o(x)=\frac{u(x)-u(-x)}{2} \quad \text{and} \quad u_e(x)=\frac{u(x)+u(-x)}{2}.
\]
Thus $u_o$ and $u_e$ are either both solutions of \eqref{eqn:1dimDE} with the same eigenvalue, or (in the case that $u$ is purely odd or purely even), one of them is zero everywhere. Because $u(x)$ and $u(-x)$ both satisfy the boundary conditions, $u_o$ and $u_e$ will also satisfy them. Since every eigenfunction is a linear combination of its odd and even parts, it suffices to look only for even and odd eigenfunctions. We will refer to eigenvalues associated with odd and even eigenfunctions as odd and even eigenvalues, respectively.

\subsection{\bf Bijections of the plane}
In this section, we state and prove the bijections of portions of the $(\tau,\mu)$-plane, which will be used in our parameterizations of the eigenvalue curves.

In later sections, we find the general form of the eigenfunctions by factoring the eigenvalue equation \eqref{eqn:1dimDE}. The factorization depends on the sign of the eigenvalue $\mu$, and in some cases on the value of $\tau$ relative to $\mu$. These differences in factorization correspond to the three different regions (upper half-plane, sub-parabolic, and super-parabolic) of the $(\tau,\mu)$ plane. In each case, we will use the boundary conditions to precisely identify the form of the eigenfunctions.

\begin{lemma}[Bijection Lemma] \label{lemma:Bij}
The following functions are bijective transformations on the indicated sets.
\begin{enumerate}
\item \label{lemma:UHP} (Upper half-plane) The function 
\[
F_{1}(a, b) = (-a^2 + b^2, a^2b^2)
\]
 maps $\{(a, b) : a, b \geq 0\}$ onto $\{(\tau, \mu) : \tau \in\RR, \mu\geq0\}$.
\item \label{lemma:LHP} (Super-parabolic region)
 The function 
\[
F_{2}(a, b) = (-a^2 - b^2, -a^2b^2)
\]
 maps  $\{(a, b) : a \geq b > 0\}$ onto $\{(\tau, \mu) : \tau<0, -\tau^2/4 \leq \mu<0\}$.
\item \label{lemma:crit} (Sub-parabolic region) The function 
\[
F_{3}(a, b) = (-2a^2 + 2b^2, -(a^2 + b^2)^2)
\]
maps $\{(a, b) : a> b>0\}$ onto $\{(\tau, \mu) : \tau< 0, \mu < -\tau^2/4\}$.
\end{enumerate}
\end{lemma}

\begin{proof} We construct explicit inverses of our functions as follows:
\begin{align*}
F_1^{-1}(\tau, \mu) &= \Big(\sqrt{\frac{\sqrt{\tau^2 + 4\mu} - \tau}{2}}, \sqrt{\frac{\sqrt{\tau^2 + 4\mu} + \tau}{2}}\Big),\\
F_2^{-1}(\tau, \mu) &= \Big(\frac{\sqrt{|\tau| + 2\sqrt{|\mu|}} + \sqrt{|\tau| - 2\sqrt{|\mu|}}}{2}, \frac{\sqrt{|\tau| + 2\sqrt{|\mu|}} - \sqrt{|\tau| - 2\sqrt{|\mu|}}}{2}\Big),\\
F_3^{-1}(\tau, \mu) &= \Big(\frac{\sqrt{-\tau + \sqrt{\tau^2 + |\tau^2 + 4\mu|}}}{2}, \frac{\sqrt{\tau + \sqrt{\tau^2 + |\tau^2 + 4\mu|}}}{2}\Big).
\end{align*} 
Note that $F_1^{-1}$ is well-defined on the region $\{(\tau, \mu) : \tau \in\RR, \mu\geq0\}$, the function $F_2^{-1}$ is well-defined on the region $\{(\tau, \mu) : \tau< 0, -\tau^2/4 \leq \mu<0\}$, and $F_3^{-1}$ is well-defined on the region $\{(\tau, \mu) : \tau< 0, \mu < -\tau^2/4\}$. Hence all three functions are indeed bijections of the appropriate sets.
\end{proof}

\begin{remark} It is obvious from the Rayleigh quotient \eqref{1dimRQ} for the free rod that there is no negative eigenvalue when $\tau\geq0$. That is, there are no eigenvalues in the fourth quadrant. For the sake of completeness, however, one could treat the fourth quadrant $\{(\tau, \mu) : \tau\geq0, \mu\leq0\}$ in the same way as the other regions. In this case, we would use the bijective transformation 
\[
F_{4}(a, b) = (a^2 + b^2, -a^2b^2), 
\]
which maps $\{(a, b) : a\geq b \geq 0\}$ onto $\{(\tau, \mu) : \tau \geq0, \mu\leq0\}$. We could then show there is no eigenvalue pair in the fourth quadrant using methods similar to those in the next section.
\end{remark}

\section{\bf The upper half-plane}\label{sec:UHP}
In this section, we treat the case of eigenvalue branches lying in the upper half-plane $\{(\tau,\mu) : \tau\in\RR,\mu\geq 0\}$. Recall that the eigenvalue equation has the form $u^{(4)}-\tau u''=\mu u$; then the characteristic equation is $r^4-\tau r^2 -\mu=0$. As we will see, the upper half-plane corresponds to the case that the characteristic equation has real and purely-imaginary complex roots. We will identify the eigenfunctions, provide a complete description for the eigenvalues as functions of $\tau$ via parameterization, and identify some key properties of the eigenvalue curves.

\subsection{\bf Eigenfunctions and eigenvalue conditions}
The starting point for solving the eigenvalue equation is factoring the eigenvalue equation \eqref{eqn:1dimDE}. When $\mu$ is non-negative, regardless of the value of $\tau$, we may factor the eigenvalue equation as
\begin{equation}\label{eqn:factorUHP}
\left(\frac{d^2}{dx^2}+a^2\right)\left(\frac{d^2}{dx^2}-b^2\right)u=0,
\end{equation}
where $\mu = a^2b^2$ and $\tau = b^2-a^2$. We may take $a$ and $b$ to be nonnegative since $\mu\geq0$.

\begin{lemma}[Eigenfunctions and eigenvalue conditions] \label{lemma:eigenfunctionsUHP}
For all eigenvalues $\mu>0$ and all $\tau\in\RR$, one of the following must hold:
\begin{enumerate}
\item The eigenvalue $\mu$ is associated with an odd eigenfunction $u_o$ of the form $u_o(x) = A \sin(ax) + B \sinh(bx)$, where $A$ and $B$ are nonzero constants, and $a$ and $b$ are nonnegative numbers such that $\mu=a^2b^2$, $\tau=b^2-a^2$, and
\begin{equation}\label{eqn:UHPevalCondOdd}
a^3\tan a = b^3\tanh b.
\end{equation}
\item The eigenvalue $\mu$ is associated with an even eigenfunction $u_e$ of the form
$u_e(x) = C \cos(ax) + D \cosh(bx)$, where $C$ and $D$ are nonzero constants, and $a$ and $b$ are nonnegative numbers such that $\mu=a^2b^2$, $\tau=b^2-a^2$, and 
\begin{equation}\label{eqn:UHPevalCondEven}
- a^3\cot a = b^3\coth b.
\end{equation}
\end{enumerate}
\end{lemma}

\begin{lemma}[Zero eigenvalues] \label{lemma:eigenfunctionsUHP2}
For all $\tau\in\RR$, the constant function $u_e(x)\equiv 1$ is an even eigenfunction with eigenvalue $\mu=0$. 

The eigenvalue $\mu=0$ has multiplicity two under the following conditions on $\tau$:
\begin{enumerate}
\item $\tau=0$. In this case, the second eigenfunction is $u_o(x)=x$ and is odd.
\item $\tau=-(l\pi)^2$, any $l\in\NN$. In this case, the second eigenfunction is $u_o(x)=\sin(ax)$ and is odd.
\item $\tau=-(2l-1)^2\pi^2/4$, any $l\in\NN$. In this case, the second eigenfunction is $u_e(x)=\cos(ax)$ and is even.
\end{enumerate}
For all other values of $\tau$, the eigenvalue $\mu=0$ is simple.
\end{lemma}

Although we state the $\mu=0$ case as a separate lemma, we treat both cases $\mu>0$ and $\mu=0$ in a single proof below.

\begin{proof} It is easy to see that the constant function $u_e\equiv 1$ satisfies our boundary value problem for all $\tau$ and has an associated eigenvalue $\mu=0$. We look now for non-constant solutions.

Since we are considering $\mu\geq0$ and $\tau\in\RR$, by Lemma~\ref{lemma:Bij}(\ref{lemma:UHP}), we may factor the eigenvalue equation as \eqref{eqn:factorUHP}. The characteristic equation then becomes  
\[
r^4 + (a^2-b^2)r^2 - a^2b^2 = 0. 
\] 
Since $a, b\geq 0$, the above quartic equation has solutions $r = \pm ia, \pm b$. We must consider four cases, depending on the positivity of $a$ and $b$.

\emph{Case 1: $a>0$ and $b>0$.} In this case, $\mu=a^2b^2$ is positive and the differential equation has four linearly independent solutions: $e^{\pm iax}$ and $e^{\pm bx}$. Because we need consider only odd and even solutions, we express these solutions instead as the trigonometric functions $\sin(ax)$ and $\cos(ax)$ and hyperbolic trigonometric functions $\sinh(bx)$ and $\cosh(bx)$. Our possible solutions are then linear combinations of these, chosen according to symmetry.

Writing $u_o$ for the odd eigenfunction and $u_e$ for the even, we have
\begin{align*}
u_o(x) &= A \sin(ax) + B \sinh(bx),\\
u_e(x) &= C \cos(ax) + D \cosh(bx).
\end{align*}
Then by Lemma~\ref{lemma:Bij}(\ref{lemma:UHP}), the boundary conditions can be expressed in terms of $a$ and $b$ as follows:
\[
\begin{cases}
u'' = 0 &\text{when $x=\pm 1$,}\\
u''' + (a^2-b^2)u'  = 0 &\text{when $x = \pm 1$.}
\end{cases}
\] 

We consider the odd eigenfunction first. We wish to determine which choices of $a$ and $b$ (and hence $\tau$ and $\mu$) yield a solution to the boundary value problem. Applying the two boundary conditions yields
\[
\begin{cases}
Aa^2\sin a - Bb^2\sinh b = 0,\\
Aab^2\cos a - Ba^2b\cosh b = 0.
\end{cases}
\]
We require that our linear combination coefficients $(A,B)$ be nontrivial, so the system's determinant must vanish. That is $-a^4b\sin a\cosh b + ab^4\cos a\sinh b = 0$. Since $a$ and $b$ are nonzero, this is equivalent to
\[
a^3\tan a = b^3 \tanh b.
\]
This gives us a condition on $(a,b)$ that assures us of an odd solution to the eigenvalue problem.

For the even eigenfunction $u_e$, applying the boundary conditions gives us
\[
 \begin{cases}
Ca^2\cos a - Db^2\cosh b = 0,\\
-Cab^2\sin a - Da^2b\sinh b = 0
\end{cases}
\]
Once again, to have nontrivial linear combination coefficients $(C,D)$, we require that the system's determinant be zero. That is, $-a^4b\cos a\sinh b - ab^4\sin a\cosh b = 0$, or equivalently, since $a$ and $b$ are nonzero,
\[
-a^3\cot a = b^3\coth b.
\]

\emph{Case 2: $a=0$ and $b>0$.} In this case, our eigenvalue $\mu=0$, and our tension parameter $\tau=b^2$ is positive. 

We also note that here $r=0$ is a double root of the characteristic equation, and so in place of $e^{\pm iax}$, the solutions we consider are  $e^{0x}$ and $xe^{0x}$. The solutions $e^{\pm bx}$  can still be expressed as hyperbolic trigonometric functions, and so when we consider the possible odd and even solutions,  we may write
\begin{align*}
u_o(x) &= A x + B\sinh(bx),\\
u_e(x) &= C + D \cosh(bx).
\end{align*}
For the odd eigenfunction $u_o$, applying the boundary conditions gives us
\[
 \begin{cases}
Bb^2\sinh b = 0,\\
Ab^2  = 0.
\end{cases}
\]
Since $b$ is positive, there is no nontrivial $(A, B)$ pair in this case, and there are no odd solutions of this form.

For the even eigenfunction $u_e$, applying the boundary conditions yields
\[
 \begin{cases}
Db^2\cosh b = 0,\\
0 = 0.  
 \end{cases}
\]
Thus we must take the coefficient $D=0$, and the only even eigenfunction is the constant solution $u_e(x)=C$.

\emph{Case 3: $a>0$ and $b=0$.} As in the previous case, our eigenvalue $\mu=0$, but this time $\tau=-a^2$ is negative.

The roots of the characteristic equation are now $r=\pm ia$ and a double root $r=0$, so for the odd and even eigenfunctions, we obtain
\begin{align*}
u_o(x) &= A \sin(ax) + Bx,\\
u_e(x) &= C \cos(ax) + D.
\end{align*}
Applying the boundary conditions to the odd eigenfunction $u_0$ yields
\[
 \begin{cases}
-Aa^2\sin a = 0,\\
a^2B = 0.  
 \end{cases}
\]
From this, we see that we must take $B=0$. Therefore, $u_o(x)$ must have the form $A\sin(ax)$ for some $A\neq0$, and the first boundary condition holds if and only if $\sin a=0$. Thus $a=l\pi$ for some natural number $l$, and $\tau=-a^2=-l^2\pi^2$.

For the even eigenfunction $u_e(x)$, applying the boundary conditions gives us
\[
\begin{cases}
-Ca^2\cos a = 0,\\
0 = 0.
\end{cases}
\]
Taking $C=0$ yields the constant eigenfunction, which we have already discussed, so we assume $C\neq0$ and instead impose the requirement $\cos a=0$. This yields $a=(2l-1)\pi/2$ for $l\in\NN$, or equivalently $\tau=-(2l-1)^2\pi^2/4$. For these values of $\tau$, we therefore have an even function $u_e(x)=C\cos(ax)$ with associated eigenvalue $\mu=0$, as desired.

\emph{Case 4: $a=0$ and $b=0$.} In this case, both our eigenvalue $\mu$ and tension parameter $\tau$ are zero. Since $r=0$ is a quadruple root of the characteristic equation, the general solution $u$ is a linear combination of $1$, $x$, $x^2$, and $x^3$. As before we only consider odd and even solutions, and write
\begin{align*}
u_o(x) &= A x + B x^3,\\
u_e(x) &= C + D x^2.
\end{align*}
For the odd eigenfunction $u_o$, applying the boundary conditions gives us the same condition for both, namely $B=0$. Therefore, $u_o(x)=Ax$ is an odd eigenfunction in this case.

Applying our boundary conditions to the even eigenfunction $u_e$ yields only $2D=0$ as a meaningful condition. Therefore, the constant function is the only even eigenfunction in this case.
\end{proof}

\subsection{\bf Parameterization of eigenvalue curves in the upper half-plane}
Lemma ~\ref{lemma:eigenfunctionsUHP} allows us to smoothly parameterize the eigenvalue branches lying in the upper half-plane.  

\begin{definition}[Parameterization for the upper half-plane] \label{defn:paramUHP}\ 
  \begin{figure}[t]
    \begin{center}
\includegraphics[scale=0.55]{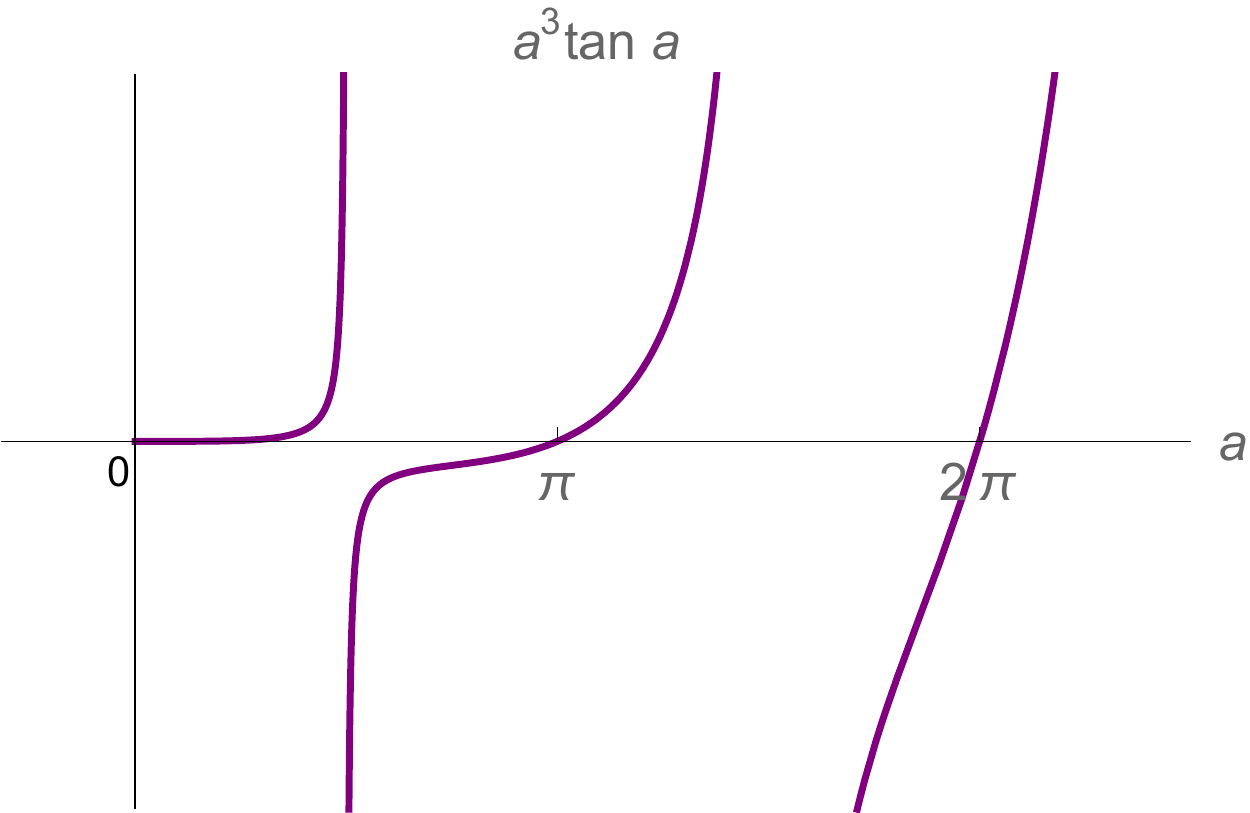}
\qquad
\includegraphics[scale=0.55]{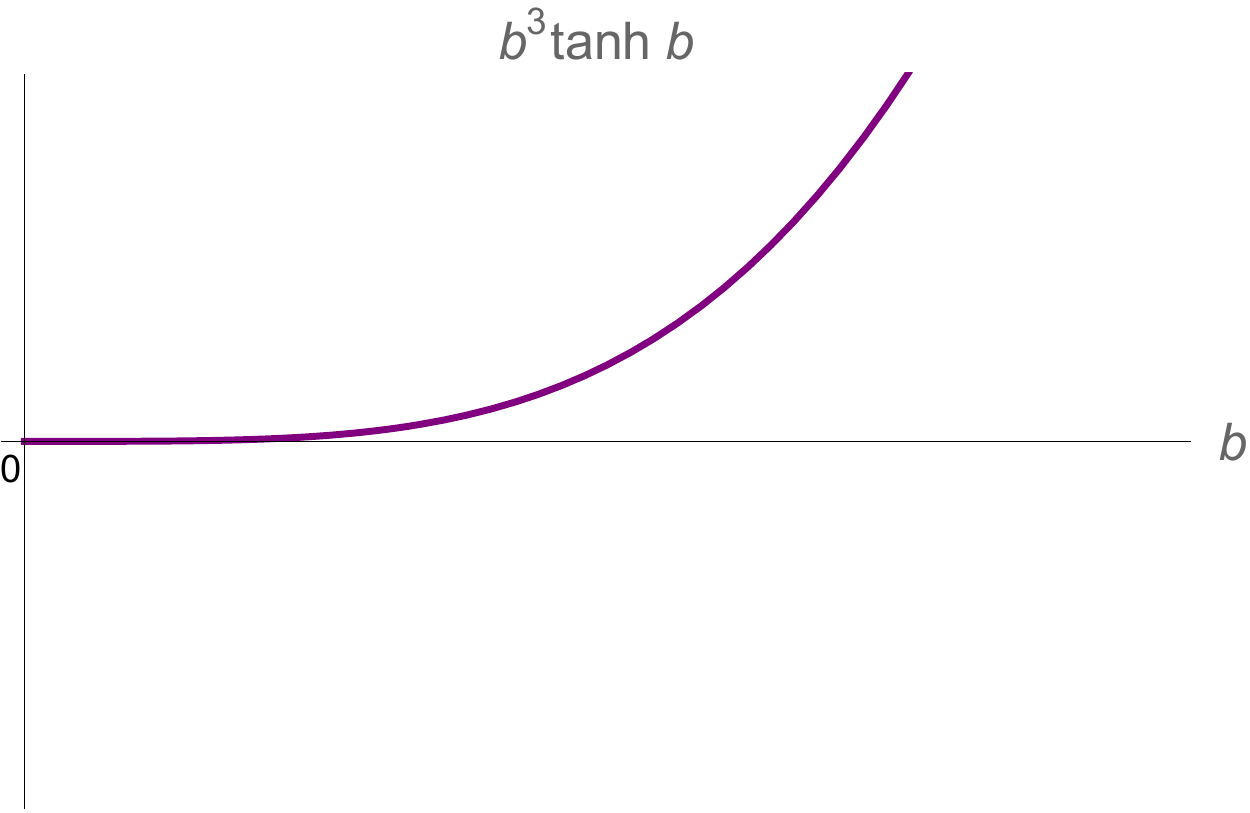}
\caption{\label{fig:UHPoddcondition} Upper half-plane eigenvalue condition $\fodd(a)=\godd(b)$ (Definition ~\ref{defn:paramUHP}) allows us to obtain $b$ in terms of $a$.}
    \end{center}
    \end{figure}
\begin{enumerate}
\item (Odd case) Define 
\[
\fodd(a) = a^3\tan a, \quad \godd(b) = b^3\tanh b, \qquad \text{for $a, b \geq 0$.} 
\]
Note that $\godd$ is increasing and hence invertible on its domain. Observe also that $\fodd$ is one-to-one when restricted to the intervals $[l\pi,(l+1/2)\pi)$ for integers $l\geq0$; this restriction is called $l$th branch of $\fodd$.

For any such $l$, we define
\[
\bodd(a) = \goddinv(a^3\tan a), \qquad  l\pi \leq a < \left(l+\frac{1}{2}\right)\pi.
\]

\item (Even case) Define 
\[
\feven(a) = -a^3\cot a,\quad \geven(b) = b^3\coth b, \qquad \text{for $a, b \geq 0$.}
\]
The function $\geven$ is increasing and hence invertible(and nonnegative). As before, we identify branches of $\feven$ by restricting its domain; since $\geven$ is positive, we consider only the positive branches of $\feven$, which are $[(l+1/2)\pi,(l+1)\pi)$ for integers $l\geq 0$. We then define
\begin{align*}
\beven(a) = \geveninv(-a^3\cot a), \qquad \left(l+\frac{1}{2}\right)\pi \leq a< (l+1)\pi.
\end{align*} 
  \begin{figure}[t]
    \begin{center}
\includegraphics[scale=0.55]{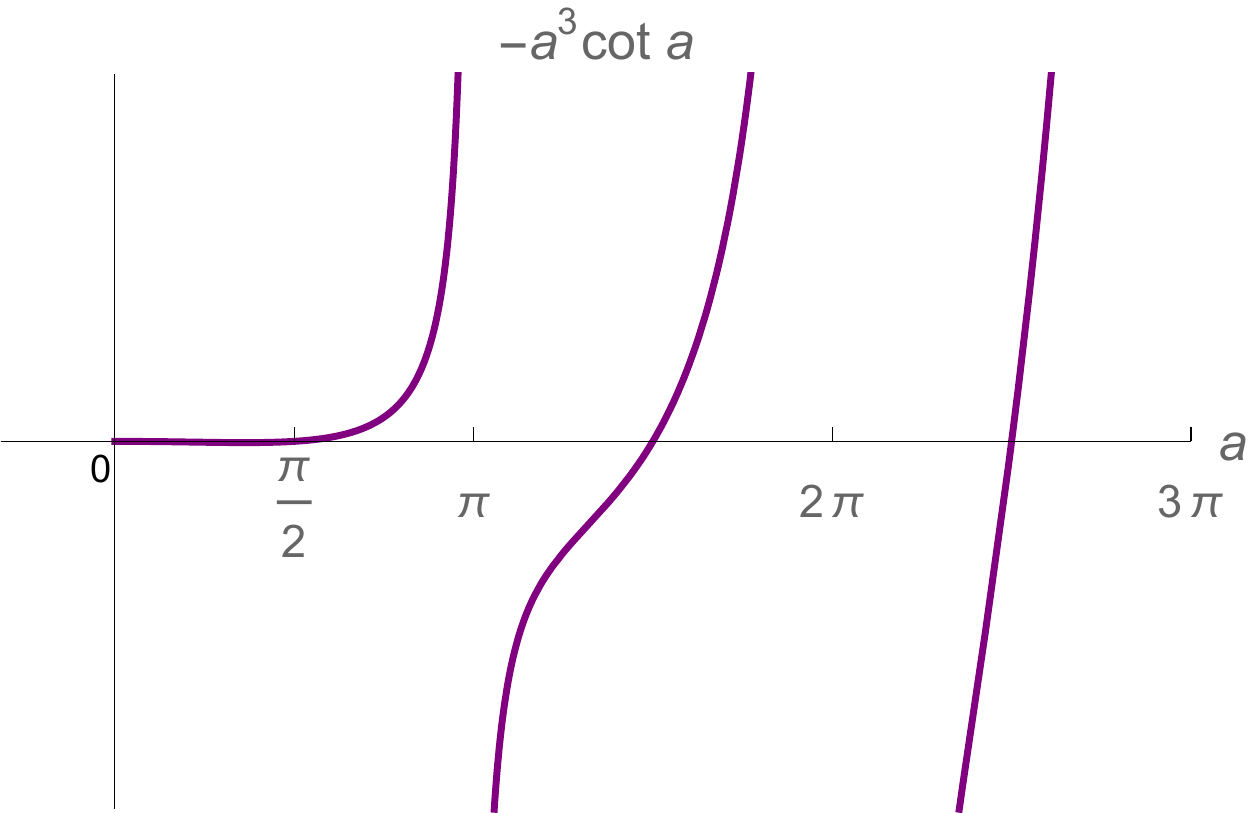}
\qquad
\includegraphics[scale=0.55]{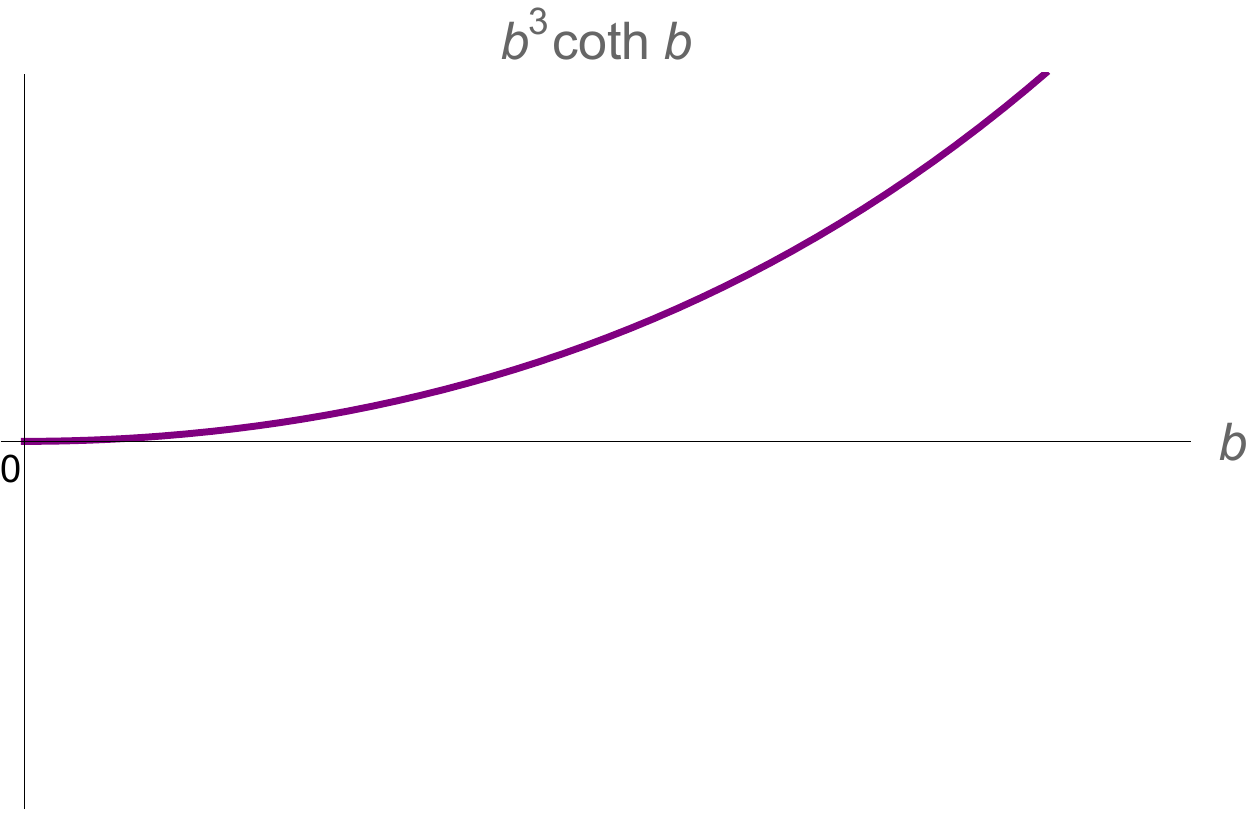}
\caption{\label{fig:UHPevencondition} Upper half-plane eigenvalue condition $\feven(a)=\geven(b)$ (Definition ~\ref{defn:paramUHP}) allows us to obtain $b$ in terms of $a$.}
    \end{center}
    \end{figure}

\end{enumerate}

  \begin{figure}[t]
    \begin{center}
\includegraphics[scale=0.6]{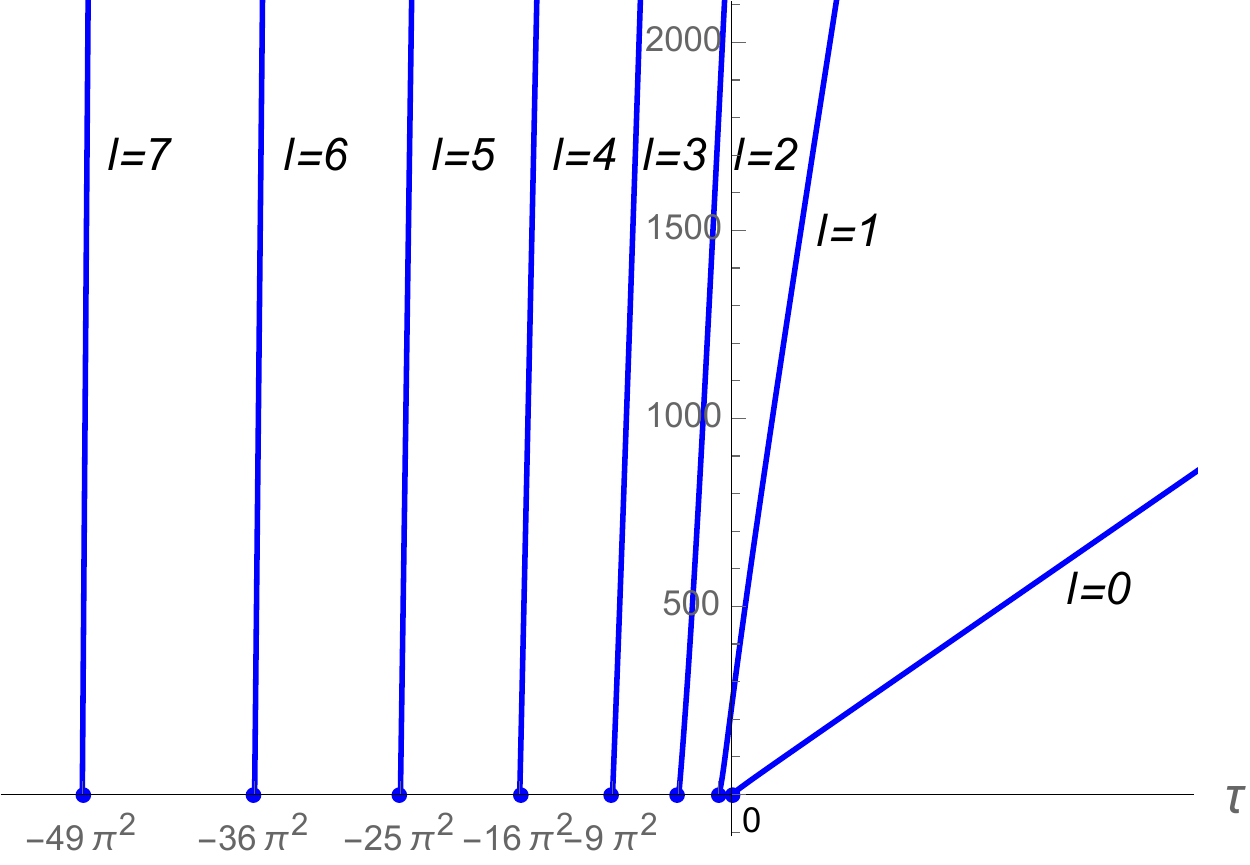}
\caption{The odd eigenvalue branches $\mu_l^{\mathrm{odd}}(\tau)$ in the upper half-plane, for $l=0, \dots,7$ (Theorem ~\ref{thm:paramUHP}).}
 \label{fig:UHPodd}
    \end{center}
    \end{figure}
\end{definition}

\begin{theorem}[Eigenvalues of the upper half-plane] \label{thm:paramUHP}
The eigenvalue curves in the upper half-plane are indexed by $l\geq 0$. For each $l$ there are two curves, according to whether the eigenfunction is odd or even: 
\begin{enumerate}
\item Odd (Figure~\ref{fig:UHPodd}): The eigenvalue curve $(\tau, \mu_l^\mathrm{odd}) = F_1(a, \bodd(a))$ is parameterized by $a\in [l\pi, (l+1/2) \pi)$.
\item Even (Figure~\ref{fig:UHPeven}): The eigenvalue curve $(\tau, \mu_l^{\mathrm{even}}) =  F_1(a, \beven(a))$ is parameterized by $a\in [(l+1/2)\pi, (l+1)\pi)$.

  \begin{figure}[t]
    \begin{center}
\includegraphics[scale=0.6]{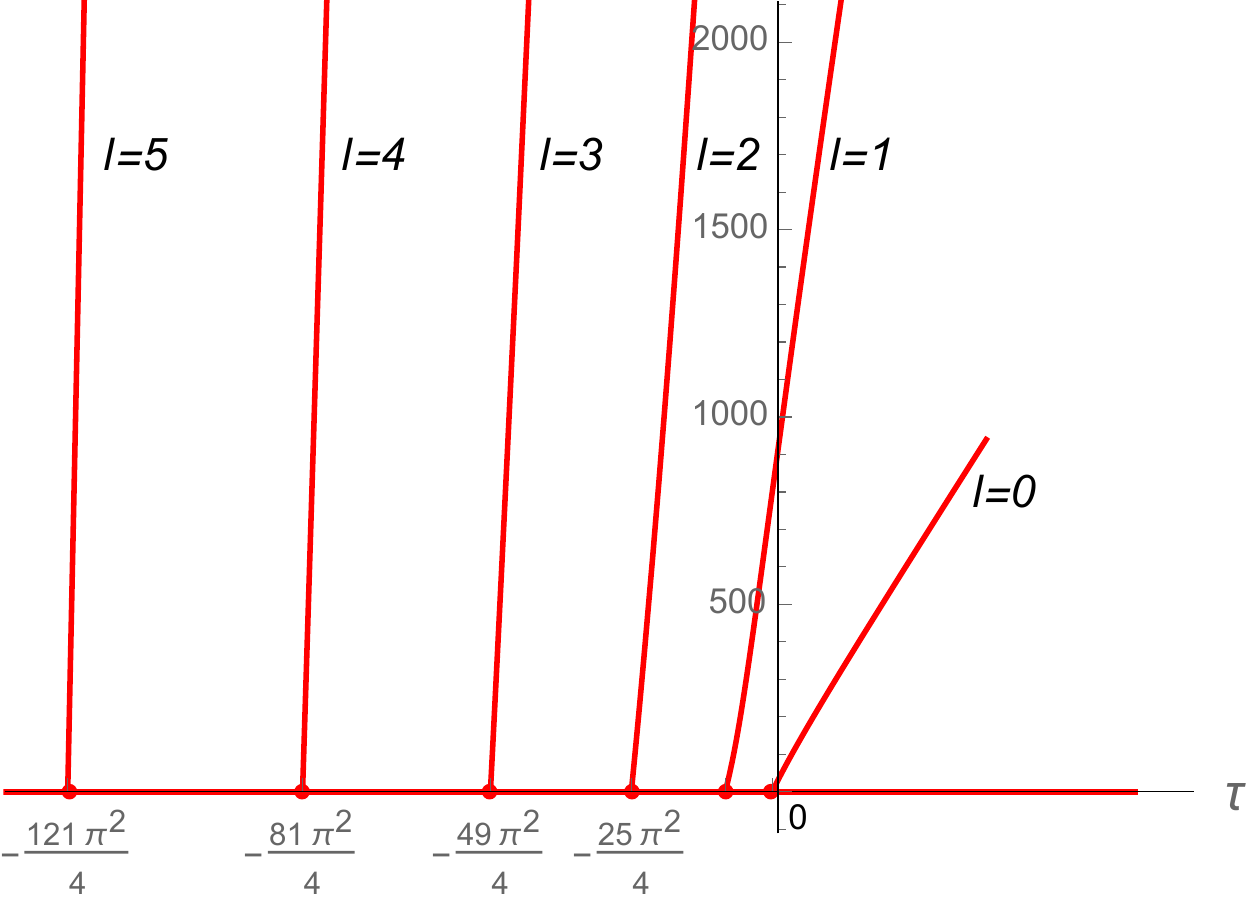}
\caption{The even eigenvalue branches $\mu_l^\mathrm{even}(\tau)$ in the upper half-plane, for $l=0, \dots,5$  and the additional zero branch $\mu \equiv0$ (Theorem ~\ref{thm:paramUHP}).}
 \label{fig:UHPeven}
    \end{center}
    \end{figure}

\item Zero eigenvalue: For all $\tau\in\RR$, the eigenvalue $\mu=0$ has constant eigenfunction (which can be regarded as a translational mode of the rod).
\end{enumerate}
\end{theorem}

\begin{proof} This theorem follows immediately from our work for Lemma~\ref{lemma:eigenfunctionsUHP},  ~\ref{lemma:eigenfunctionsUHP2} and our definitions in Definition~\ref{defn:paramUHP}.
\end{proof}

\subsection{\bf Properties of the eigenvalue curves in the upper half-plane}
In this section, we state and prove several properties of the eigenvalue branches lying in the upper half-plane. We also state some properties which are clear from numerical investigation but for which we do not have rigorous proof. Our first results, about intersections of eigenvalue branches, follow from our parameterizations.

\begin{proposition}[Nonintersection of Eigenvalue Curves]\label{prop:UHPintersection} The only intersection of eigenvalue branches in the upper half-plane occurs at the origin. More precisely, for all indices $l_1,l_2$, we have
\begin{align*}
\mu_{l_2}^\mathrm{odd}(\tau) &>\mu_{l_1}^{\mathrm{odd}}(\tau) \qquad\text{when $l_2>l_1\geq 0$ and $\tau\geq-(l_1 \pi)^2$,}\\
\mu_{l_2}^{\mathrm{even}}(\tau) &>\mu_{l_1}^{\mathrm{even}}(\tau) \qquad\text{when $l_2>l_1\geq 0$ and $\tau\geq-\left((l_1+1/2) \pi\right)^2$,}\\
\mu_{l_2}^\mathrm{even}(\tau) &>\mu_{l_1}^{\mathrm{odd}}(\tau) \qquad\text{when $l_2\geq l_1\geq 0$ and $\tau\geq-(l_1 \pi)^2$,}\\
\mu_{l_2}^{\mathrm{odd}}(\tau) &>\mu_{l_1}^{\mathrm{even}}(\tau) \qquad\text{when $l_2>l_1\geq 0$ and $\tau\geq-\left((l_1+1/2) \pi\right)^2$.}
\end{align*}
\end{proposition}

\begin{proof} By Lemma~\ref{lemma:Bij}(\eqref{lemma:UHP}), since the $(a,b)$ and $(\tau,\mu)$ pairs are in one-to-one correspondence, we cannot have different $(a, b)$ pairs mapped to the same $(\tau, \mu)$. Thus the only possible intersections occur when different-symmetry eigenvalues have the same $(a,b)$ values.

Suppose the odd and even eigenvalue curves intersect. At this point, $a$ and $b$ satisfy eigenvalue conditions \eqref{eqn:UHPevalCondOdd} and \eqref{eqn:UHPevalCondEven} simultaneously: 
\begin{align*}
a^3\tan a &= b^3\tanh b,\\
-a^3\cot a &= b^3\coth b.
\end{align*}
When $b>0$, we can conclude that both $\tan a$ and $\cot a$ are finite and so multiplication of equations gives $-a^6=b^6$, which has no solution. When $b=0$, the conditions $a^3\tan a=0$ and $-a^3\cot a=0$ are simultaneously satisfied only when $a=0$. Therefore, \eqref{eqn:UHPevalCondOdd} and \eqref{eqn:UHPevalCondEven} both hold only when $(a, b)=(0,0)$, which corresponds to $(\tau, \mu)=(0,0)$.
\end{proof}

\begin{proposition}[Direction of parameterization] 
For a fixed $l\geq 0$, the parameterizations of $\mu_l^{\mathrm{odd}}(\tau)$ and $\mu_l^{\mathrm{even}}(\tau)$ in Theorem~\ref{thm:paramUHP} travel to the right and upwards. That is, $\tau$ and $\mu$ are strictly increasing functions of $a$.
\end{proposition}

\begin{proof} For the purpose of this proof, we will show that the eigenvalue curves parameterized according to Theorem~\ref{thm:paramUHP} are exactly those given by Poincar\'e's min-max characterization
\begin{equation*}
\mu_j = \min_{S_j}\max_{u\in S_j} \frac{\int_{-1}^{1} |u''|^2 + \tau |u'|^2 \, dx}{\int_{-1}^{1} |u|^2 \, dx},
\end{equation*} 
where $S$ ranges over all $j$-dimensional subspaces of $H^2\left((-1, 1)\right)$. Observe that the eigenvalue $\mu_j$ is increasing as a function of $\tau\in\RR$.

Until now, we have used the terms ``eigenvalue curve'' and ``eigenvalue branch'' interchangeably. For the duration of this proof, we will use the former to refer to the parameterized curves, and denote them by $\mu_l^{\text{odd}} (\tau)$ and $\mu_l^{\mathrm{even}}(\tau)$ for $l\geq0$. The $j$th eigenvalue branch will mean the set of pairs $(\tau,\mu_j)$ obtained from the Poincar\'e principle. Our goal is then to show that each eigenvalue branch corresponds with one of our parameterized curves.

We will prove this invertly, that is, we consider $\tau$ as a function of $\mu$. Then the eigenvalue curves can be thought of as graphs of functions $\tau_l^{\mathrm{odd}}(\mu)$ and $\tau_l^{\mathrm{even}}(\mu)$ for $l\geq0$.

For any free rod eigenvalue $\mu$ associated with a nonconstant eigenfunction $u$, the eigenvalue equation \eqref{eqn:1dimDE} with free boundary conditions is satisfied for $(\tau,\mu)$ if and only if $-\tau$ is an eigenvalue for the equation
\begin{equation} \label{eqn:DEbuckling}
u^{(4)} - \mu u = -(-\tau) u^{\prime\prime}
\end{equation} 
with boundary conditions arising naturally from the minimizers of the associated Rayleigh quotient (for $-\tau$), which has the form 
\begin{equation*}
R[v] = \frac{\int_{-1}^{1} |v''|^2 - \mu |v|^2 \, dx}{\int_{-1}^{1} |v'|^2 \, dx}.
\end{equation*} 
We take  $v\in H^2\left((-1, 1)\right)$ such that $\int_{-1}^{1} v\, dx=0$, so that $v$ is not a constant function. Note that \eqref{eqn:DEbuckling} is the well-known ``buckling eigenvalue'' problem.

Fix an index $l$ and consider the odd eigenvalue curve parameterized by $a\in\left[l\pi, (l+1/2)\pi\right)$. Recall that our parameterization allows us to consider $b$ as a function of $a$. We know from the parameterization that $\mu=a^2b^2$ depends smoothly on $a$ with derivative
\begin{equation*}
\frac{d\mu}{da} = 2ab^2 + 2a^2b\frac{db}{da}.
\end{equation*}
We also know $\mu=0$ when $a=l\pi$ and that $\mu\to\infty$ as $a\to(1+1/2)\pi$. 

From implicit differentiation of the odd eigenvalue condition \eqref{eqn:UHPevalCondOdd}, we obtain
\[
\frac{db}{da} = \frac{3a^2\tan a + a^3\sec^2a}{3b^2\tanh b + b^3(1-\tanh^2b)}.
\]
The denominator is always positive since $\tanh b <1$ for all $b>0$. The numerator is positive for $a\in[l\pi,(l+1/2)\pi)$, and therefore, $db/da$ is positive for all $a$ under consideration. Thus $d\mu/da >0$, and so we may regard the curve as being parameterized by $\mu\in[0, \infty)$. The argument is similar for the even eigenvalue curves. 

Now we write the eigenvalue curves as graphs in the $(\mu,\tau)$-plane of the functions
\[
\tau = \tau_l^{\mathrm{odd}}(\mu)\qquad\text{or}\qquad \tau = \tau_l^{\mathrm{even}}(\mu),
\]
with $\mu\in[0, \infty)$. From Proposition~\ref{prop:UHPintersection}, we know the ordering of $\tau_l^{\mathrm{odd}}(0)$'s and $\tau_l^{\mathrm{even}}(0)$'s for each $l \geq 0$. In other words,
\begin{align*}
\tau_l^{\mathrm{odd}}(0) > \tau_l^{\mathrm{even}}(0) > \tau_{l+1}^{\mathrm{odd}}(0) > \tau_{l+1}^{\mathrm{even}}(0) > \dots
\end{align*}
Again from Proposition~\ref{prop:UHPintersection}, the eigenvalue curves for nonconstant eigenfunctions do not intersect and the ordering at $\mu=0$ is maintained. That is,
\[
\tau_l^{\mathrm{odd}}(\mu) > \tau_l^{\mathrm{even}}(\mu) > \tau_{l+1}^{\mathrm{odd}}(\mu) > \tau_{l+1}^{\mathrm{even}}(\mu) >  \dots  \quad \forall \mu\in[0, \infty).
\]
Therefore, these curves are in fact the eigenvalue branches for the buckling eigenvalues $\tau_1(\mu), \tau_2(\mu), \dots$ of \eqref{eqn:DEbuckling}. 
So we have 
\begin{align*}
\tau_0^{\mathrm{odd}} &= \tau_1,\\
\tau_0^{\mathrm{even}} &= \tau_2,\\
 &\text{etc.}
\end{align*} 
The Rayleigh quotient for \eqref{eqn:DEbuckling} tells us that the $j$th eigenvalue $-\tau_j(\mu)$ is decreasing as a function of $\mu$ for each $j$, and so $\tau_j(\mu)$ is increasing as a function of $\mu$. Hence $\tau_l^{\mathrm{odd}}(\mu)$ and $\tau_l^{\mathrm{even}}(\mu)$ are increasing as functions of $\mu$. They are in fact strictly increasing; otherwise, a single $\tau$-value would correspond to a whole interval of $\mu$-values solving the free rod boundary value problem, which is impossible since the spectrum is discrete. Therefore, each eigenvalue curve is strictly increasing for $\tau$ as a function of $\mu$. Hence by inverting to get $\mu$ as a function of $\tau$, we see that each eigenvalue curve $\mu_l^{\mathrm{odd}}(\tau)$ and $\mu_l^{\mathrm{even}}(\tau)$ are strictly increasing as functions of $\tau$.
\end{proof}

We also list a number of properties of the eigenvalue branches which can be observed numerically, but which we do not prove rigorously. We state only the odd case, since the arguments are similar for the even case.

\begin{enumerate}
\item \textbf{As $l$ increases, the corresponding eigenvalue branch lie further to the left.} 

This can easily be seen in Figures~\ref{fig:UHPodd} and \ref{fig:UHPeven}, or by Proposition~\ref{prop:UHPintersection}.

\item \textbf{Vertical intercepts of parameterized eigenvalue curves.} The vertical intercept of the $l$th eigenvalue branch occurs at $(0, a^4)$, for some $a$. As $l\to\infty$, we have $a = l\pi + \pi/4 + o(1)$.

This can be justified algebraically as follows. Vertical intercepts occur when $\tau=0$, and hence $a = b$. We want values of $a$ such that $a^3\tan a = a^3\tanh a$, and so we must have $a=0$ or $\tan a = \tanh a$. The above asymptotic for $a$ can be improved with better control on how quickly $\tanh(a)\to1$ as $a\to\infty$.

\item \textbf{The eigenvalue curves become linear in limiting cases.}
\begin{enumerate}
\item  As $l$ increases, the eigenvalue curve looks more like to a straight line. This tendency is more pronounced as $l$ increases.
\begin{proof} (Sketch.) For large values of $l$, the graph of $y=a^3\tan a$ becomes quite steep at $a=l\pi$. If $a$ is increased some small amount $\varepsilon$, from $l\pi$ to $l\pi + \varepsilon$, the corresponding $b$ value (see Figure~\ref{fig:UHPoddcondition}) increases by a large amount and thus $\tau = -a^2 + b^2$ also increases greatly. Recall $\mu$ can be expressed in terms of $a$ and $\tau$ as
\[
\mu = a^2\tau + a^4 \qquad\text{ where $\tau = -a^2 + b^2$}.
\]
Note that $\tau$ is growing quickly while $a^2=(l\pi)^2+O(\varepsilon)$ and $a^4=(l\pi)^4+O(\varepsilon)$ remain relatively constant, and hence $\mu$ is approximately linear as a function of $\tau$ for this small range of $a$ values, until $a$ is close to $(l+1/2)\pi$. 
\end{proof}

\item As $\tau$ tends to $\infty$, the eigenvalue curves converge to straight lines with slopes corresponding to the eigenvalues of a free vibrating string.

The connection between the free rod and free string can most intuitively be seen by considering the Rayleigh quotients. The Rayleigh quotient for the free string is given by 
\[
 Q_s[u]=\frac{\int_{-1}^1|u'|^2\,dx}{\int_{-1}^1u^2\,dx}.
\]
If we divide the rod Rayleigh quotient $Q$ (from \eqref{1dimRQ}) by $\tau$, the result can be written as
\[
 \frac{Q[u]}{\tau}=\frac{1}{\tau}\frac{\int_{-1}^1|u''|^2\,dx}{\int_{-1}^1u^2\,dx}+ Q_s[u].
\]
Note that for a fixed function $u$ and large values of $\tau$, the string Rayleigh quotient dominates. We thus expect the slopes of the eigenvalue curves to approach the eigenvalues of the free string. Although there are complications with the spaces over which we take the infima, this almost-linear relationship between the eigenvalues of the string and the rod can be made rigorous in the case of the first nonzero eigenvalue in all dimensions (see \cite{Chasman11}).

In just one dimension, we can make a stronger case by investigating the boundary value problem directly: 
\begin{proof}
Recall $\tau = -a^2 + b^2$ and $\mu = a^2b^2=a^2\tau + a^4$. Fix $l$, so that we consider only the $l$th branch of the eigenvalue curve and of the $a^3\tan a$ graph. As $a$ increases from $l\pi$ to $(l+1/2)\pi$, the corresponding $b$ also increases, and hence $\tau$ does as well. 

Changing our perspective, as $\tau\to\infty$ (and hence $a^3\tan(a)\to\infty$ along the $l$th branch), the value of $a$ satisfying the eigenvalue condition approaches $(l+1/2)\pi$. Thus we can consider $a$ to be nearly constant for large $\tau$, and so we see that $\mu= a^2\tau+a^4$ is nearly linear in $\tau$, with approximate slope $a^2\approx (l+1/2)^2\pi^2$.
\end{proof}
\end{enumerate}
\end{enumerate}

\section{\bf The lower half-plane: super-parabolic region}\label{sec:LHP}
In this section, we address the case of negative eigenvalues whose curves lie above the critical parabola $\mu = -\tau^2 /4$ for all $\tau <0$. We identify the eigenfunctions and derive the eigenvalue conditions from the natural boundary conditions. We also provide a complete description of the eigenvalues as functions of $\tau$ via parameterization, and identify some key properties of the eigenvalue curves.

\subsection{\bf Eigenvalue conditions in the super-parabolic region}
 As we will see, the region $\{(\tau, \mu) : \tau<0, -\tau^2/4\leq \mu <0\}$ corresponds to the characteristic equation $r^4-\tau r^2 -\mu=0$ having purely imaginary roots. When $\mu$ and $\tau$ are both negative and satisfy $\mu \geq -\tau^2 /4$, we may factor the eigenvalue equation as
\begin{equation}\label{eqn:factorLHP}
\left(\frac{d^2}{dx^2}+a^2\right)\left(\frac{d^2}{dx^2}+b^2\right)u=0,
\end{equation}
where $\mu = -a^2b^2$ and $\tau = -a^2-b^2$ by Lemma~\ref{lemma:Bij}(\ref{lemma:LHP}). We may take $a$ and $b$ to be positive since $\mu < 0$.

\begin{lemma}[Eigenfunctions and eigenvalue conditions] \label{lemma:eigenfunctionsLHP}
For all $\tau<0$ and eigenvalues $\mu$ satisfying $-\tau^2 /4 \leq \mu <0$, at least one of the following must hold: 
\begin{enumerate}
\item The eigenvalue $\mu$ is associated with an odd eigenfunction $u_o$ of the form $u_o(x) = A \sin(ax) + B \sin(bx)$, where $A$ and $B$ are nonzero constants, and $a$ and $b$ are positive numbers such that $\mu=-a^2b^2$, $\tau = -a^2-b^2$, and 
\begin{equation}\label{eqn:LHPevalCondOdd}
a^3\tan a = b^3\tan b.
\end{equation}
\item The eigenvalue $\mu$ is associated with an even eigenfunction $u_e$ of the form $u_e(x) = C \cos(ax) + D \cos(bx)$, where $C$ and $D$ are nonzero constants, and $a$ and $b$ are positive numbers such that $\mu=-a^2b^2$, $\tau = -a^2-b^2$, and 
\begin{equation}\label{eqn:LHPevalCondEven}
 a^3\cot a = b^3\cot b.
\end{equation}
\item The eigenvalue $\mu$ is associated with an even eigenfunction $u_e(x) = C \cos(ax) + Dx \sin(ax)$, where $C$ and $D$ are nonzero constants, and $a\approx 1.13943$ satisfies
\[
 \sin(2a) = \frac{2a}{3}.
\]
This is the only eigenvalue satisfying $\mu=-\tau^2/4$ and occurs when $\tau\approx-2.5966$ and $\mu\approx-1.6856$. 
\end{enumerate}
\end{lemma}

\begin{proof} Since we are considering those $(\tau,\mu)$ pairs satisfying $-\tau^2/4\leq\mu<0$ and $\tau<0$, we may factor the eigenvalue equation as \eqref{eqn:factorLHP}. By Lemma ~\ref{lemma:Bij}(\ref{lemma:LHP}), the boundary conditions can be expressed in terms of $a$ and $b$ as follows:
\[
\begin{cases}
u'' = 0 &\text{when $x=\pm 1$,}\\
u''' + (a^2+b^2)u' = 0 &\text{when $x = \pm 1$.}
\end{cases}
\] 
From the factorization~\eqref{eqn:factorLHP}, the characteristic equation is 
\[
r^4 + (a^2+b^2)r^2 + a^2b^2 = 0. 
\]
Since $a, b > 0$, the quartic equation has solutions $r = \pm ia, \pm ib$. We must consider two cases, depending on the multiplicities of the roots.

\emph{Case 1: $a\neq b$.}
In this case, the differential equation has four linearly independent solutions: $e^{\pm iax}$ and $e^{\pm ibx}$. Because we have shown that we need consider only odd and even solutions, we express these solutions instead as linear combinations of trigonometric functions, chosen according to symmetry.

Writing $u_o$ for the odd eigenfunction and $u_e$ for the even eigenfunction, we have
\begin{align*}
u_o(x) &= A \sin(ax) + B \sin(bx),\\
u_e(x) &= C \cos(ax) + D \cos(bx).
\end{align*}
We consider the odd eigenfunction first. We wish to determine which choices of $a$ and $b$ (and hence $\tau$ and $\mu$) yield a solution to the boundary value problem. Applying the two boundary conditions yields
\[
\begin{cases}
-Aa^2\sin a - Bb^2\sin b = 0,\\
-Aab^2\cos a - Ba^2b\cos b = 0.
\end{cases}
\]
We require that our linear combination coefficients $(A, B)$ be nontrivial, so the system's determinant must vanish. Since $a$ and $b$ are nonzero, this is equivalent to
\[
a^3\tan a = b^3 \tan b.
\]
For the even eigenfunction $u_e$, applying the boundary conditions gives us
\[
\begin{cases}
-Ca^2\cos a - Db^2\cos b = 0,\\
Cab^2\sin a + Da^2b\sin b = 0,
\end{cases}
\]
and reasoning as before, we conclude that $a, b$ must satisfy
\[
a^3\cot a = b^3\cot b.
\]

\emph{Case 2: $a=b$.} In this case, we now have $\tau=-2a^2$ and $\mu=-a^4$, and thus $\mu=-\tau^2/4$. Therefore, this case corresponds to points $(\tau,\mu)$ on the critical parabola.

Our characteristic equation becomes $r^4 + 2a^2r^2 + a^4 = 0$, which has two double purely-imaginary roots $r=\pm ia$. Expressing the linearly independent solutions in terms of trigonometric functions, we have $\sin(ax)$, $\cos(ax)$, $x\sin(ax)$, and $x\cos(ax)$.  We now have the odd and even eigenfunctions
\begin{align*}
u_o(x) &= A \sin(ax) + Bx \cos(ax),\\
u_e(x) &= C \cos(ax) + Dx \sin(ax).
\end{align*}
Applying our boundary conditions to the odd eigenfunction, we obtain
\[
\begin{cases}
-Aa^2\sin a - 2Ba\sin a - Ba^2\cos a = 0,\\
-Aa^3\cos a + Ba^2\cos a + Ba^3\sin a = 0.
\end{cases}
\]
Again we require that our linear combination coefficients be nontrivial, so the system's determinant must vanish. Thus $-3a^3\sin a\cos b - a^4 = 0$, and since $a$ is nonzero, 
\[
\sin(2a) = -\frac{2a}{3}.
\]
This equation has no real solutions for $a>0$, and so there is no odd solution in this case.

For the even eigenfunction $u_e$, applying the boundary conditions gives us
\[
\begin{cases}
-Ca^2\cos a + 2Da\cos a - Da^2\sin a = 0,\\
Ca^3\sin a + Da^2\sin a - Da^3\cos a = 0.
\end{cases}
\]
This time, the requirement that $(C,D)$ be nontrivial yields $-3a^3\sin a\cos b + a^4 = 0$, or equivalently, since $a$ is nonzero,
\[
\sin(2a) = \frac{2a}{3}.
\]
This equation has only one positive solution, $a\approx 1.13943$. We find our approximate $\tau$ and $\mu$ values according to $\tau=-2a^2$ and $\mu=-a^4$. Thus we have only one eigenvalue on the critical parabola, corresponding to the above even eigenfunction.
\end{proof}

\subsection{\bf Parameterization of eigenvalue curves in the super-parabolic region}

Lemma~\ref{lemma:eigenfunctionsLHP} allows us to smoothly parameterize the eigenvalue branches lying in the super-parabolic region, mirroring our approach for the upper half-plane. We treat odd and even branches separately. 

\subsubsection{\bf Parameterization of odd eigenvalue branches}\label{sec:paramSPOdd}
Recall that our odd eigenvalue condition \eqref{eqn:LHPevalCondOdd} can be written as $a^3\tan a=b^3\tan b$. With this in mind, we define the functions
\begin{equation*}
\fodd(a) = a^3\tan a,\qquad \godd(b) = b^3\tan b, \qquad a, b>0
\end{equation*}
Our odd eigenvalues are thus determined by the equation $\fodd(a)=\godd(b)$, and we wish to express this condition explicitly in the form $a=\foddinv(\godd(b))$. However, the function $\fodd$ is not one-to-one on $a>0$, and we have infinitely many choices of restricted domain. This is how our parameterization produces the infinitely many eigenvalue branches that we observe numerically (see, e.g., Figure~\ref{fig:zoomspec}).

\begin{figure}[t]
    \begin{center}
\includegraphics[scale=0.55]{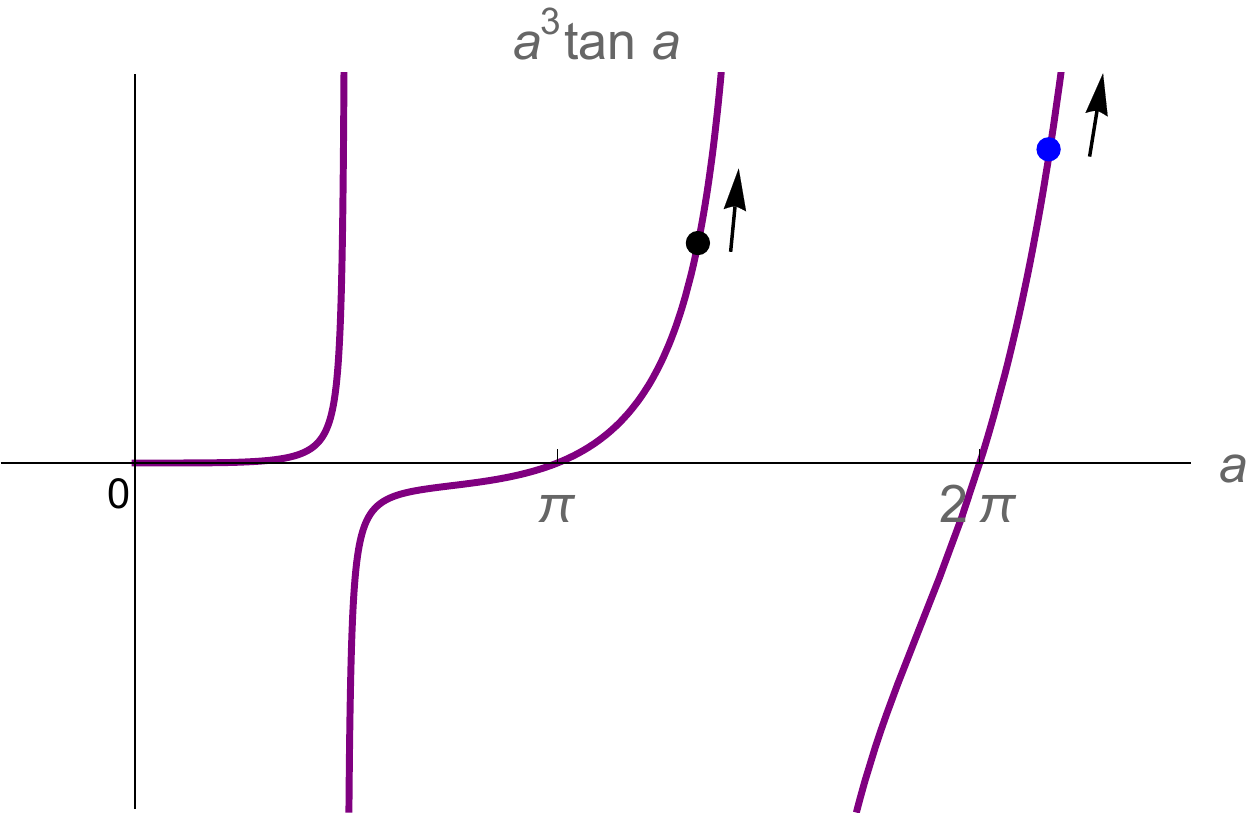}
\qquad
\includegraphics[scale=0.55]{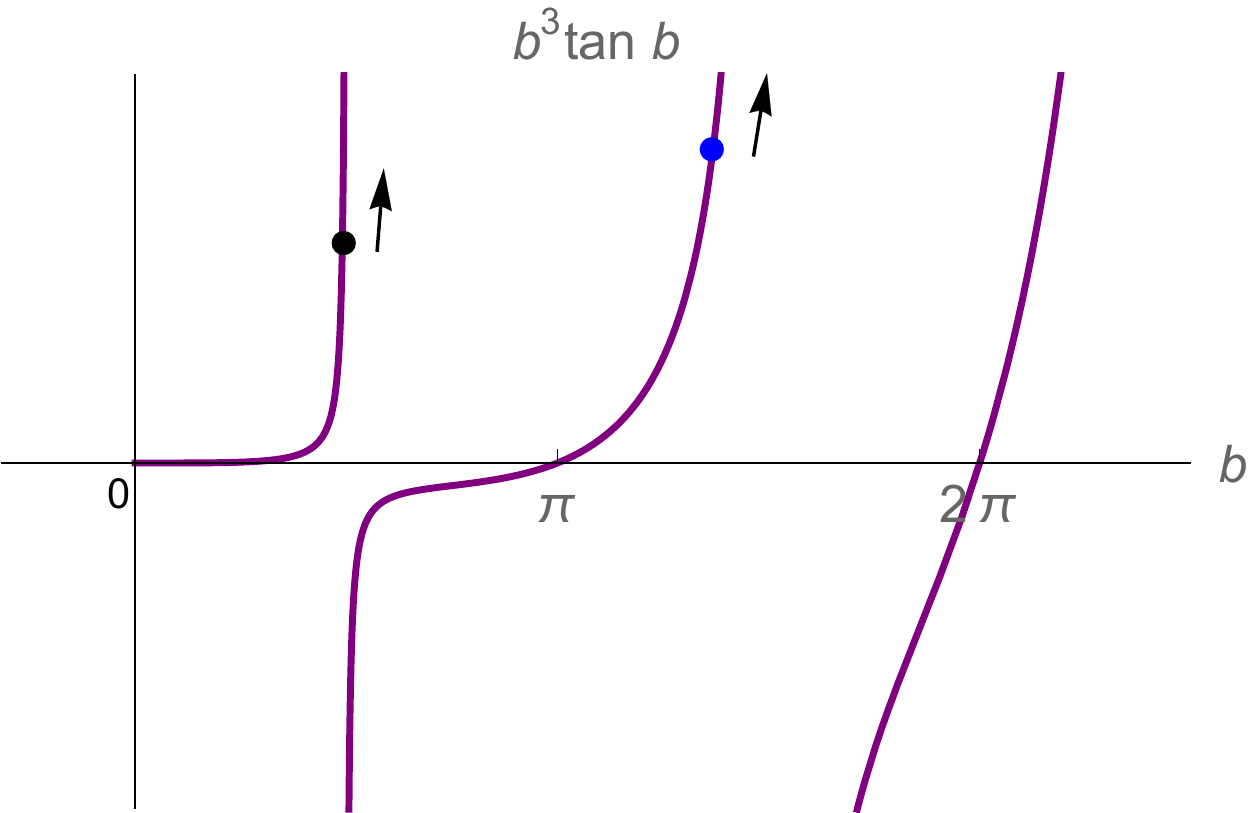}
\caption{\label{fig:LHPoddcondition} By Lemma~\ref{lemma:eigenfunctionsLHP} and the definitions of $\fodd$ and $\godd$, the odd eigenvalue condition for the super-parabolic region can be written as $\fodd(a)=\godd(b)$. The points on the above curves illustrate the dependence of $a$ on $b$ for a given branch, that is, $a=a_{l, \mathrm{odd}}(b)$. The points shown here are on the $l=1$ branch, and the arrows indicate motion of points as $a$ increases.}  
     \end{center}
    \end{figure}
    
\begin{definition}[Parameterization for the super-parabolic region, odd case]\label{defn:paramLHPOdd}
Observe that $\fodd$ is one-to-one when restricted to intervals $((j-1/2)\pi, (j+1/2)\pi)]$ for integers $j\geq1$. Write $R_j\fodd(a)$ for the restriction of $\fodd$ to its $j$th branch, that is,
\[
R_j\fodd(a) = \fodd\Big|_{\big((j-1/2)\pi, (j+1/2)\pi\big]}, \qquad j\geq 1.
\]
These functions are bijections from their domains to all of $\RR$. Therefore, for each integer $l\geq 1$, we may define a map $\foddinv(b;l)$ which maps from $b\in(0, \infty)$ to the interval $a\in(l\pi - \pi/2, \infty)$ according to
\[
\foddinv(b; l) = (R_{k+l}\fodd)^{-1}(b^3\tan b), \quad \max\left\{\left(k-\frac{1}{2}\right)\pi, 0\right\} < b\leq \left(k + \frac{1}{2}\right)\pi,\ k\geq 0.
\]
To understand this map, notice that $\foddinv(b;l)$ maps restricted domains of the $k$th branch of $\fodd$ to the $k+l$th branch, i.e.,
\begin{align*}
\left(k\pi-\frac{\pi}{2}, k\pi + \frac{\pi}{2}\right] \mapsto \left((k+l)\pi - \frac{\pi}{2}, (k+l)\pi + \frac{\pi}{2}\right].
\end{align*}
Finally, we define the composition of the restrictions of $\foddinv$ with $\godd(b)$:
\[
a_{l, \mathrm{odd}}(b) = \foddinv(b^3\tan b; l) \quad\text{for $b>0$}. 
\]
We now have a family of parameterizations for the branches given by the odd eigenvalue condition~\eqref{eqn:LHPevalCondOdd}, indexed by integers $l\geq 1$.
\end{definition}

\subsubsection{\bf Parameterization of even eigenvalue branches}
As in the odd case, our goal is to rewrite the even eigenvalue condition \eqref{eqn:LHPevalCondEven} as $-a^3\cot a=-b^3 \cot b$ and in an explicit form, providing a parameterization of the eigenvalue branches.
\begin{definition}[Parameterization for the super-parabolic region, even case]\label{defn:paramLHPEven}
Define 
\[
\feven(a) = -a^3\cot a,\quad \geven(b) = -b^3\cot b, \qquad a, b>0.
\]
We also define the numbers $a^{\ast}$ (resp. $b^\ast$) to be the point where the first branch of $\feven$ (resp. $\geven$) obtains its minimum (see Figure~\ref{fig:LHPevenfirstcondition}). Observe that $\feven$ is one-to-one when restricted to intervals $(j\pi, (j+1)\pi]$ for integers $j\geq1$, and to $[a^{\ast}, \pi]$.

Write $R_j\feven(a)$ for the restriction of $\feven$ to its $j$th branch, that is,
\begin{align*}
R_j\feven(a) &= \feven\Big|_{\big(j\pi, (j+1)\pi\big]},\qquad j\geq 1,\\
R_0\feven(a) &= \feven\Big|_{\big[a^{\ast}, \pi\big]}.
\end{align*}
Since these restrictions are invertible, we define the maps $\feveninv(b;l)$ for $l\geq 1$, which map $b\in(0, \infty)$ to $a\in(l\pi, \infty)$ according to
\[
\feveninv(b; l) = (R_{k+l}\feven)^{-1}(-b^3\cot b) \quad k\pi < b\leq (k+1)\pi,\ k\geq 0.
\]
To understand this family of maps, notice that $\feveninv(b;l)$ maps
\begin{align*}
\Big(k\pi, (k+1)\pi\Big] \mapsto \Big((k+l)\pi, (k+l+1)\pi\Big].
\end{align*}
Likewise, define the map $\feveninv(b;0)$ which maps from $b\in(0, b^{\ast}]$ to $a\in[a^{\ast}, \pi/2)$ by
\[
\feveninv(b; 0) = (R_{0}\feven)^{-1}(-b^3\cot b) \quad\text{when $b\in(0, b^{\ast}]$}.
\]
For any integer $l\geq0$, we may now express the $l$th branch of the eigenvalues according to $a=a_{l, \mathrm{even}}(b)$, where
\[
a_{l, \mathrm{even}}(b) = \feveninv(-b^3\cot b; l) \text{\quad for $b>0$}. 
\]
\end{definition}

  \begin{figure}[t]
    \begin{center}
\includegraphics[scale=0.55]{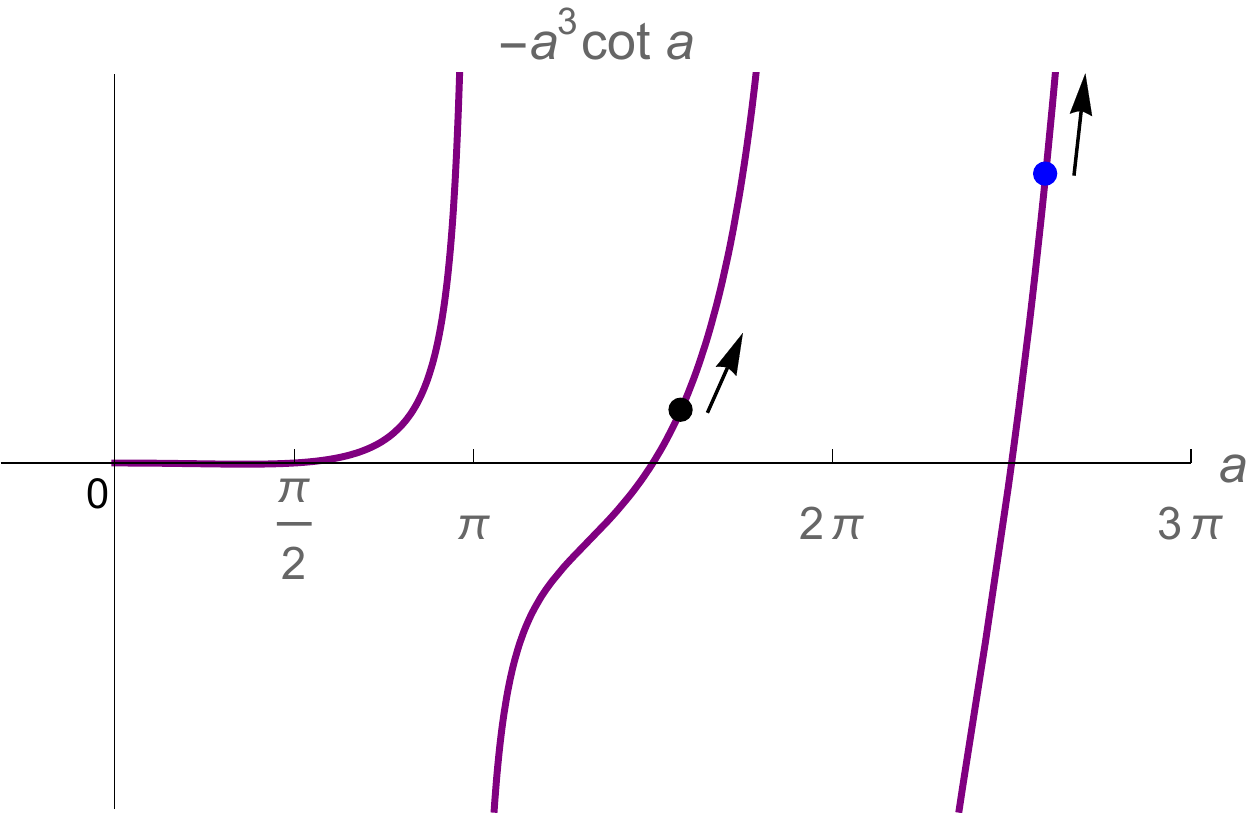}
\qquad
\includegraphics[scale=0.55]{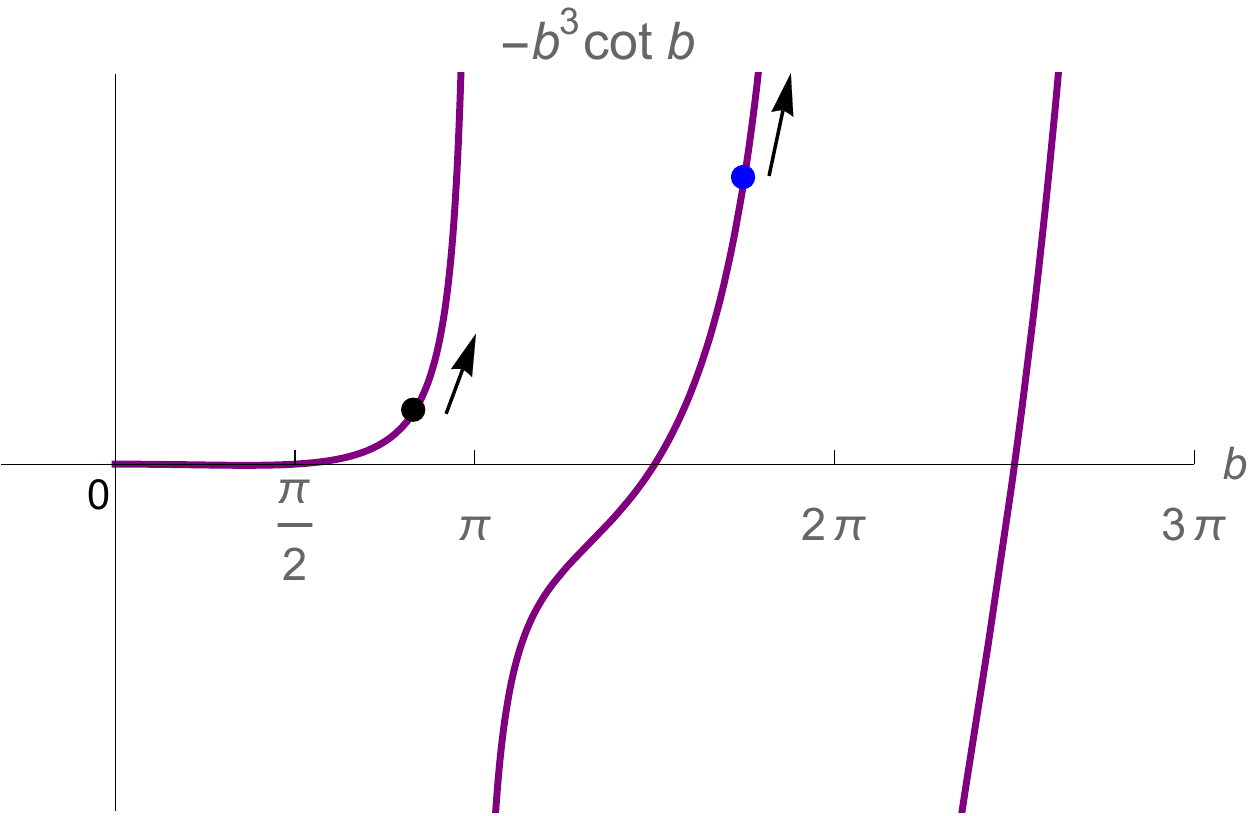}
\caption{\label{fig:LHPevencondition} By Lemma~\ref{lemma:eigenfunctionsLHP}, the even eigenvalue condition is $\feven(a)=\geven(b)$ in the super-parabolic region. As in the odd case, the points on the curves indicate how to interpret $a$ in terms of $b$ and the branch $l$, that is, $a_{l, \mathrm{even}}(b)$. The points shown correspond to the case $l=1$, and the arrows indicate motion of points as $a$ increases.}
    \end{center}
    \end{figure}

\begin{remark}\label{rmk:sameab} The eigenvalue conditions we have derived so far take the form $f(a)=g(b)$. In the upper half-plane, we parameterized the eigenvalue curves in terms of $a$, inverting $g(b)$. By contrast, we chose $b$ as our parameter for eigenvalue curves in the super-parabolic region, inverting $f(a)$. In this case, the function $f(a)=-a^3\cot a$ is not one-to-one unless we restrict our domain. 

The first branch of $\feven$ requires further restriction, taking either the portion to the left or the right of its minimum at $a=a^{\ast}$ (see Figure~\ref{fig:LHPevenfirstcondition}). Since $a$ is greater than $b$, we discard the portion of the first branch of $\feven$ on $[0, a^{\ast})$. 
\end{remark}

\begin{figure}[t]
    \begin{center}
\includegraphics[scale=0.55]{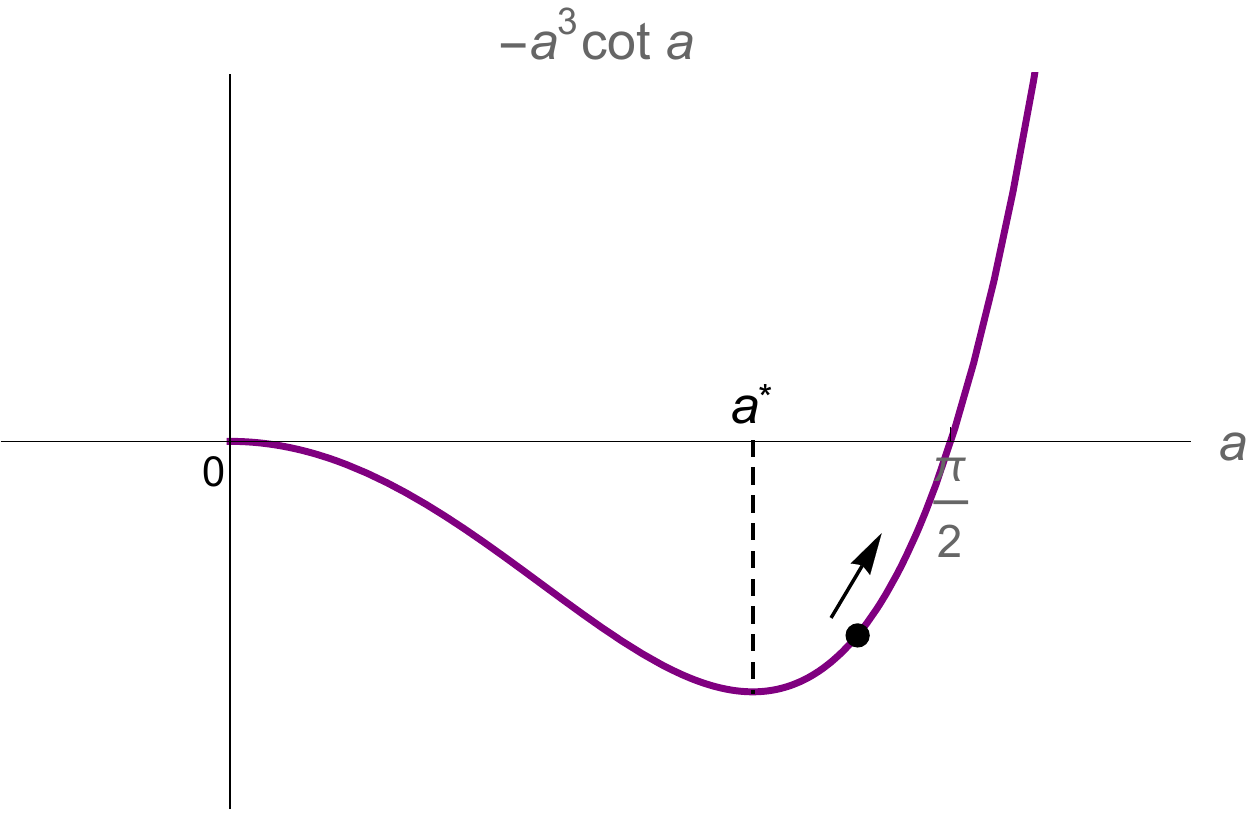}
\includegraphics[scale=0.55]{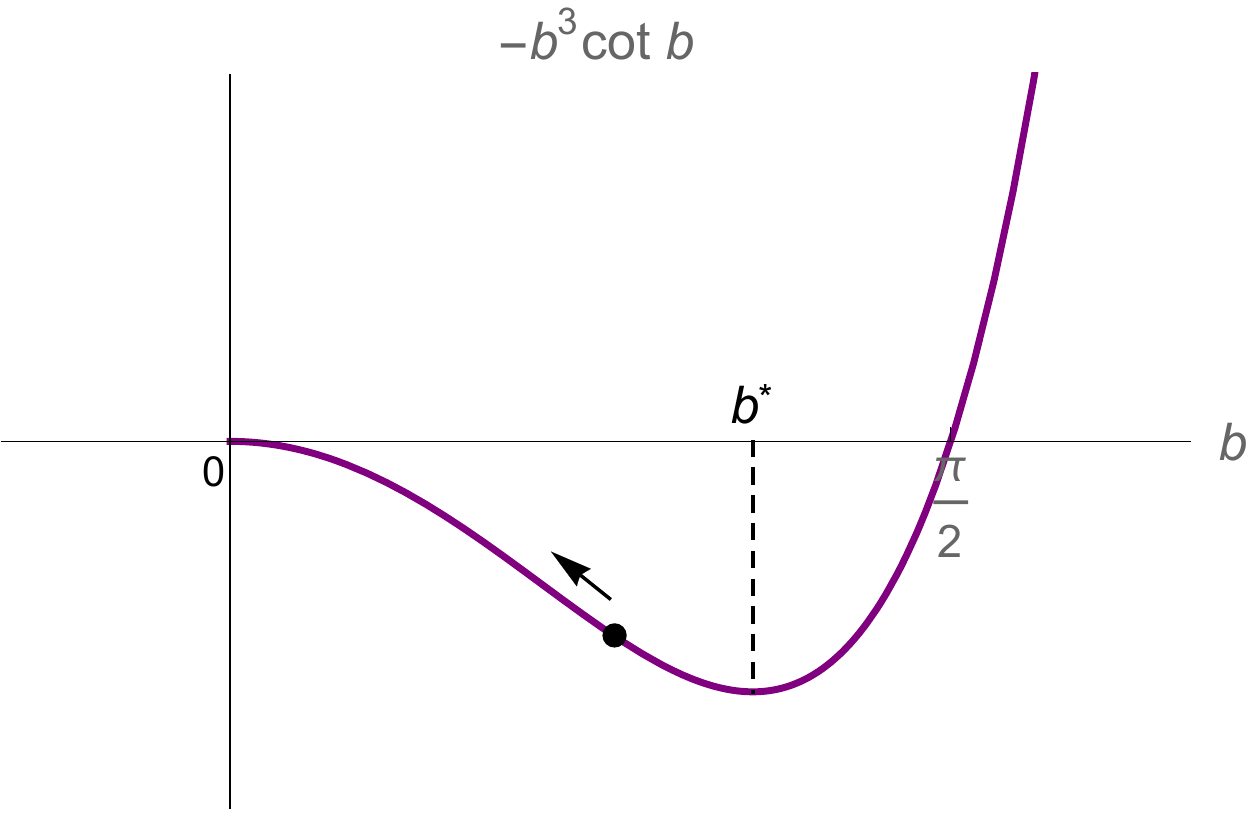}
\caption{\label{fig:LHPevenfirstcondition} Super-parabolic region: the first branch of the even eigenvalue condition $f_{\mathrm{even}}(a) = g_{\mathrm{even}}(b)$ gives $a_{0, \mathrm{even}}(b)$ (see Remark~\ref{rmk:sameab}). The arrows indicate the motion of points as $a$ increases. Here $a^\ast \leq a < \pi/2$ and $0< b\leq b^\ast$, where $a^\ast = b^\ast$ are the locations of the minima. }
    \end{center}
    \end{figure}
The case when $a$, $b$ values on the same branch gives us a little piece of the even eigenvalue curve that connects to a point on the critical parabola (Figure~\ref{fig:LHPeven(0)}).
 \begin{figure}[H]
    \begin{center}
\includegraphics[scale=0.6]{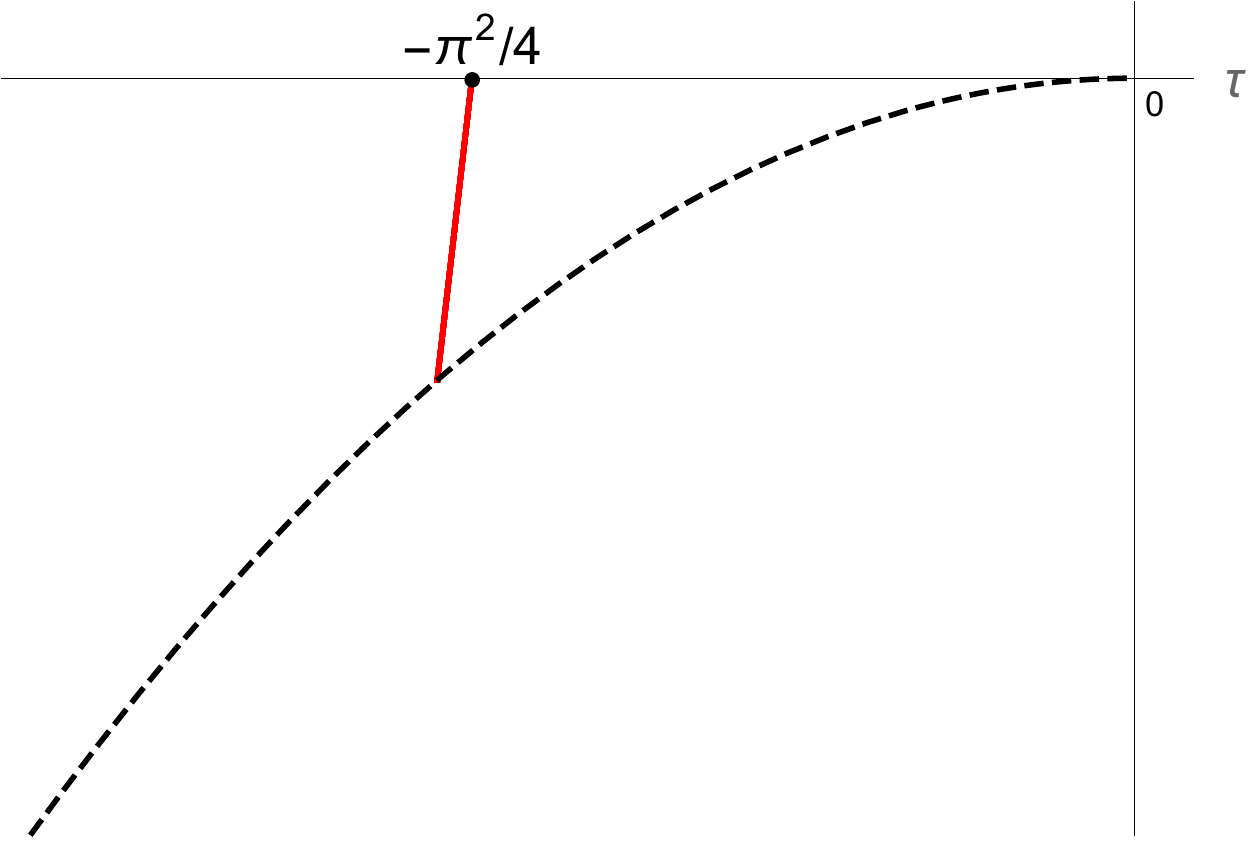}
\caption{\label{fig:LHPeven(0)} The even eigenvalue branch $\mu_0^{\mathrm{even}}(\tau)$ (Remark~\ref{rmk:sameab}), which connects $(-\pi^2/4, 0)$ on the $\tau$-axis with the critical parabola (Theorem~\ref{thm:paramLHP}(\ref{thm:paramLHPEvenA})).}
    \end{center}
    \end{figure}

\begin{theorem}[Eigenvalues of the super-parabolic region]\label{thm:paramLHP}
The eigenvalue curves in the super-parabolic region are indexed by integers $l$. For each index, there is one odd eigenvalue curve and one even eigenvalue curve.
\begin{enumerate}
\item Odd (Figure~\ref{fig:LHPodd}): For each $l\geq 1$, the eigenvalue curve $(\tau, \mu_l^{\mathrm{odd}}) =  F_2(a_{l, \mathrm{odd}}(b), b)$ is parameterized by $b > 0$. 
  \begin{figure}[t]
    \begin{center}
\includegraphics[width=2.5in]{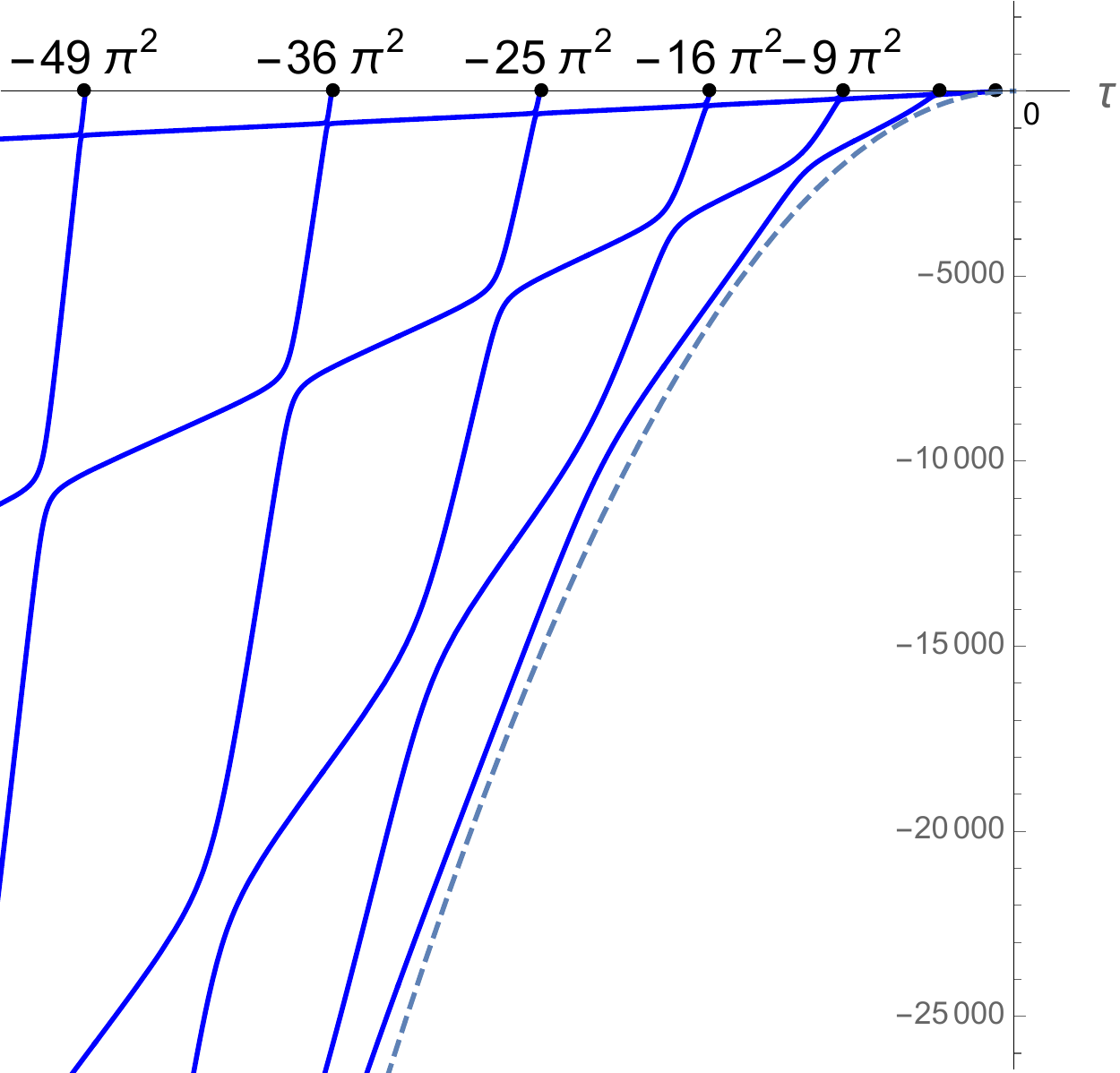}
\qquad
\includegraphics[width=2.5in]{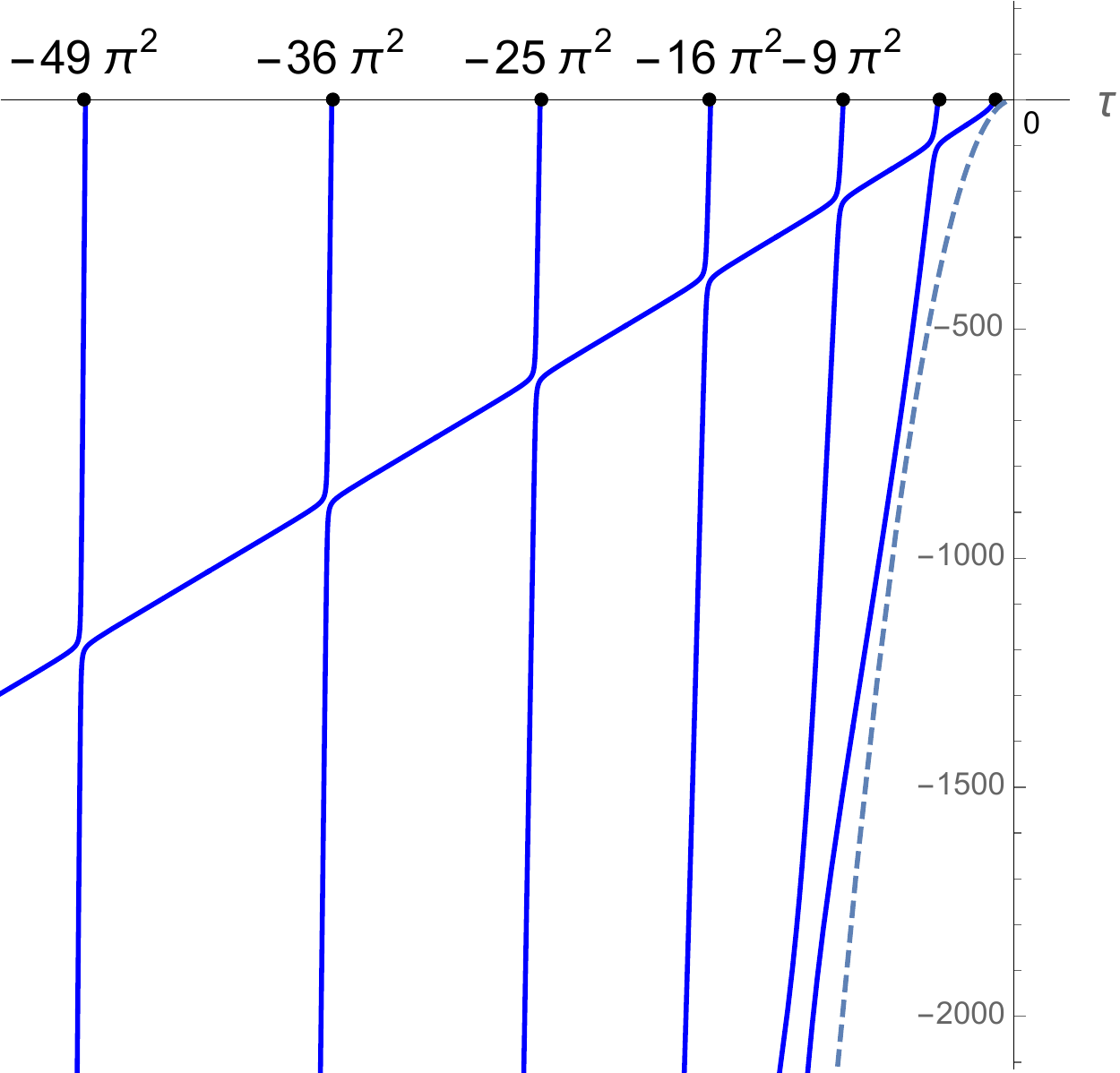}
\caption{\label{fig:LHPodd} The odd eigenvalue branches $\mu_l^{\mathrm{odd}}(\tau)$ in the lower half-plane above the critical parabola $\mu = -\tau^2 / 4$ (dashed) for $l=1,\dots,7$, and zoomed-in vertically near the $\tau$-axis (Theorem ~\ref{thm:paramLHP}). Horizontal intercepts are at $-l^2 \pi^2$.}
    \end{center}
    \end{figure}
\item Even (Figure~\ref{fig:LHPeven}): For each $l\geq 0$, the eigenvalue curve $(\tau, \mu_l^{\mathrm{even}}) = F_2(a_{l, \mathrm{even}}(b), b)$ is parameterized by $b > 0$. \label{thm:paramLHPEven}
  \begin{figure}[b]
    \begin{center}
\includegraphics[width=2.5in]{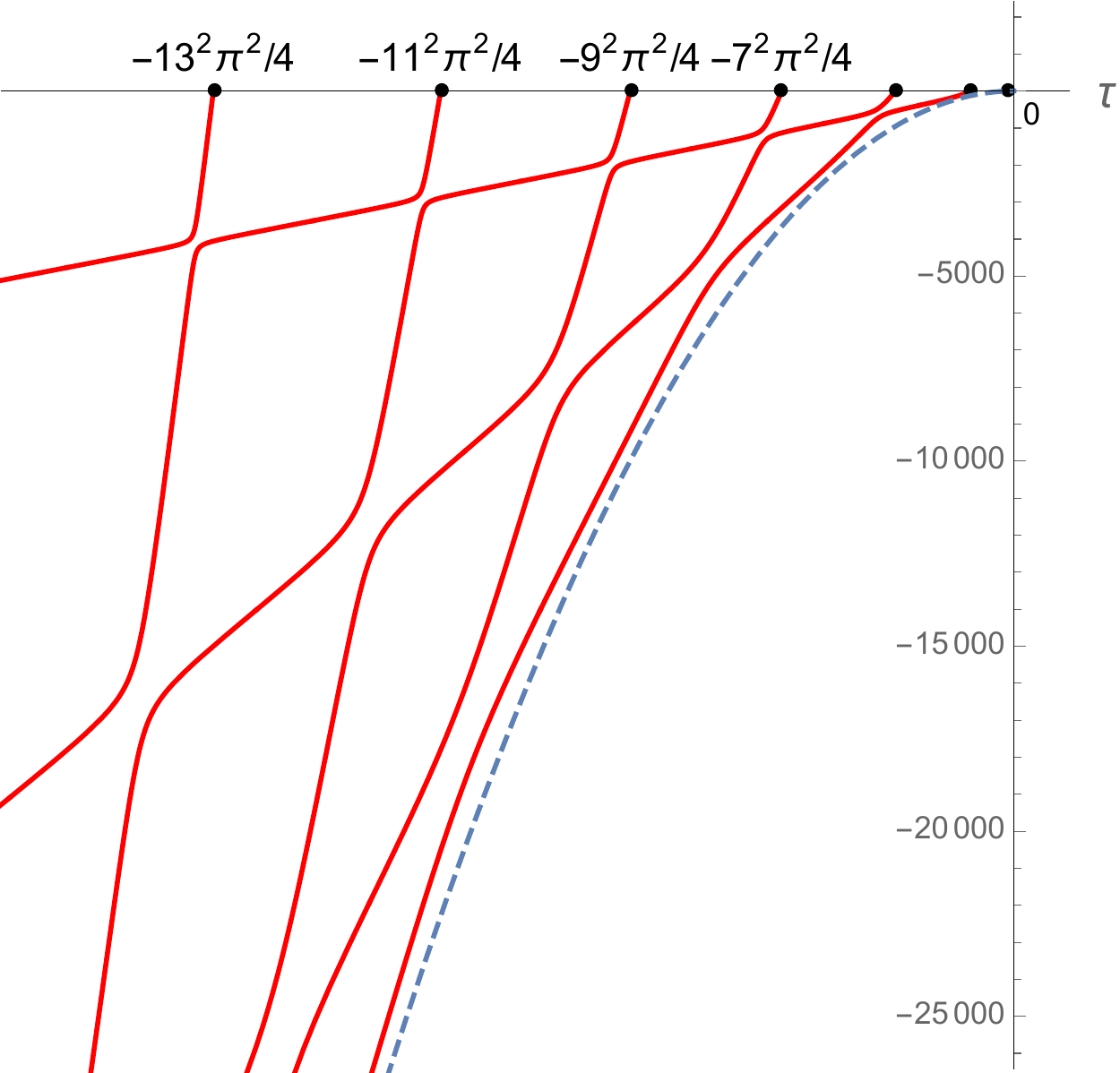}
\qquad
\includegraphics[width=2.5in]{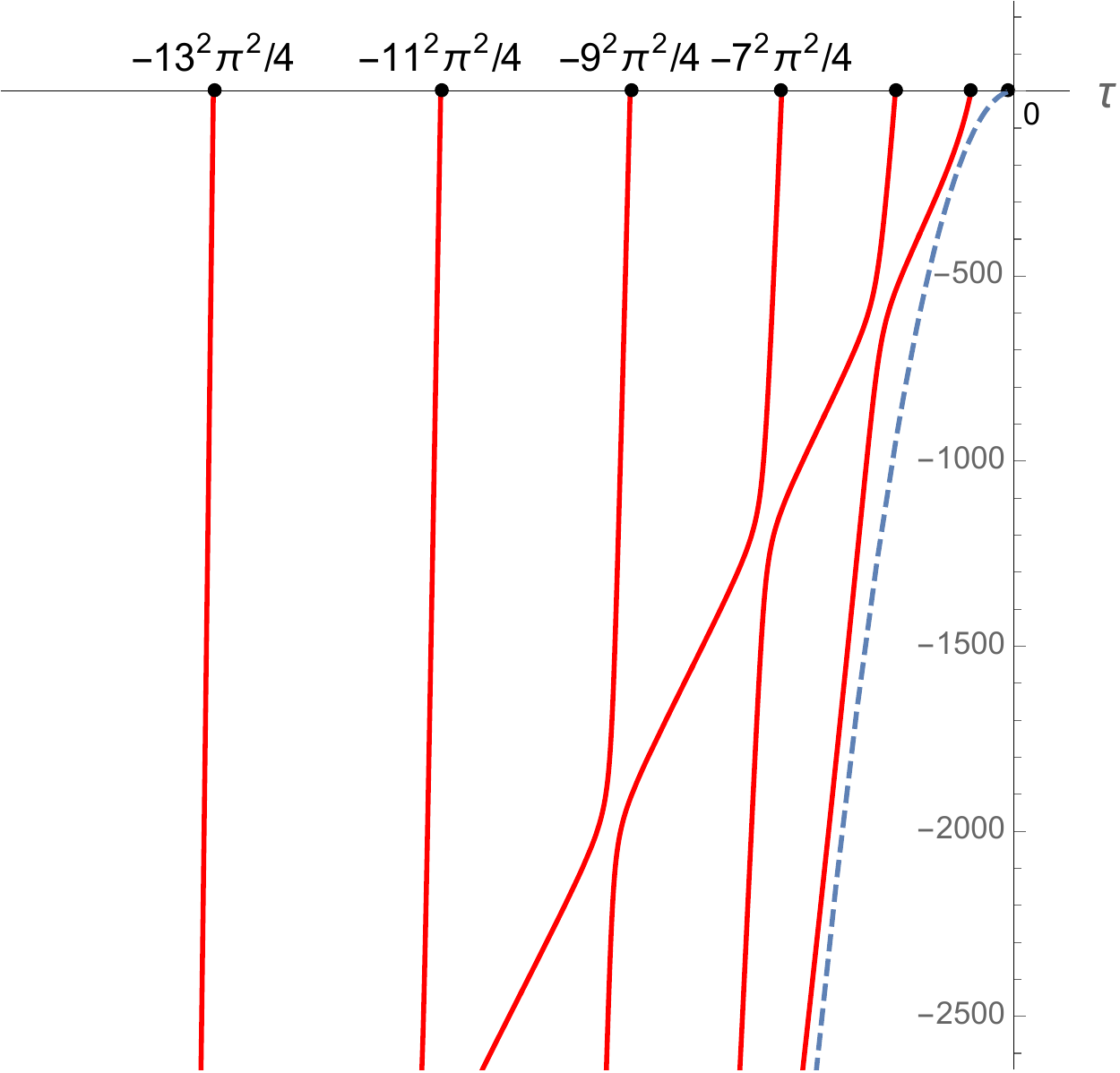}
\caption{\label{fig:LHPeven}The even eigenvalue branches $\mu_l^{\mathrm{even}}(\tau)$ in the lower half-plane above the critical parabola $\mu = -\tau^2 / 4$ (dashed) for $l=0,\dots,6$, and zoomed-in vertically near the $\tau$-axis (Theorem ~\ref{thm:paramLHP}). Horizontal intercepts are at $-(l+1/2)^2\pi^2$.}
    \end{center}
    \end{figure}
\item Critical parabola: There exists a unique $(\tau,\mu)$ pair on the critical parabola $\mu=-\tau^2/4$. This eigenvalue is $\mu = -a^4$ with corresponding even eigenfunction $u = C \cos(ax) + D x\sin(ax)$, where $a\approx 1.13943$. \label{thm:paramLHPEvenA}
\end{enumerate}
\end{theorem}

\begin{proof} This theorem follows immediately from our work for Lemma~\ref{lemma:eigenfunctionsLHP} and our Definition~\ref{defn:paramLHPOdd} and \ref{defn:paramLHPEven}.
\end{proof}

\begin{remark} The parameterization of the eigenvalue curves is quite different in the upper half-plane compared to the super-parabolic region. In the upper half-plane, each eigenvalue branch consists of a single parameterization. In the super-parabolic region, each eigenvalue branch consists of infinitely many pieces. The difference is that for the upper half-plane (Figure ~\ref{fig:UHPoddcondition} and ~\ref{fig:UHPevencondition}), the functions $\godd$ and $\geven$ are monotonic and can be inverted,  whereas in the super-parabolic region (Figure ~\ref{fig:LHPoddcondition} and ~\ref{fig:LHPevencondition}), neither $\fodd$ and $\feven$ nor $\godd$ and $\geven$ are globally invertible. Thus, we must restrict the branches to achieve local invertibility.
\end{remark}

\subsection{\bf Some properties of the eigenvalue curves in the super-parabolic region}\label{sec:propertyLHP}
In this section, we state and prove several properties of the eigenvalue branches lying in the super-parabolic region. We also state some properties which are clear from numerical investigation but for which we do not have rigorous proof.

We begin with a proposition identifying the intersections of eigenvalue branches.

\begin{proposition} \label{prop:LHPintersections}
\begin{enumerate}
\item Same-symmetry eigenvalue branches in the super-parabolic do not intersect. That is,\begin{align*}
\text{if $l_2>l_1> 0$, then}\quad \mu_{l_2}^{\mathrm{odd}}(\tau) &>\mu_{l_1}^{\mathrm{odd}}(\tau) \quad\text{for all $\tau\leq -(l_2 \pi)^2$,}\\
\text{if $l_2>l_1\geq 0$, then}\quad\mu_{l_2}^{\mathrm{even}}(\tau) &>\mu_{l_1}^{\mathrm{even}}(\tau) \quad\text{for all $\tau\leq-\left((l_2+\frac{1}{2}) \pi\right)^2$.}
\end{align*}
\item Different-symmetry eigenvalue branches in the super-parabolic region with the same index $l$ intersect infinitely often. These intersections occur at the points
\begin{equation} \label{SP:intersections}
(\tau, \mu) = \left(-\frac{(m^2 + l^2)\pi^2}{2}, -\frac{(m^2-l^2)^2\pi^4}{16}\right), \quad\text{for $m\in \NN$}.
\end{equation}
\end{enumerate}
\end{proposition}

\begin{proof}
By the bijection $F_2$ from Lemma~\ref{lemma:Bij}(\eqref{lemma:LHP}), any intersections of eigenvalue branches in the super-parabolic region in the $(\tau,\mu)$ plane are in one-to-one correspondence with intersections in the $(a,b)$ plane. For a given symmetry type, the $(a,b)$ pairs uniquely determine the eigenfunction (up to multiplication by a constant), so the only possible intersections of eigenvalue curves occur when odd and even eigenvalue curves with same the same $(a,b)$ meet.
 
Since we are looking at the same index $l$ for both even and odd branches, we must consider branches $j$ and $k$ of $\fodd$ and $\feven$, with $j-k=l$ (Section~\ref{sec:paramSPOdd}). The odd and even eigenvalue curves meet when $a$ and $b$ satisfy the eigenvalue conditions \eqref{eqn:LHPevalCondOdd} and \eqref{eqn:LHPevalCondEven} simultaneously, which means
\begin{align*}
a^3\sin a \cos b &= b^3\cos a \sin b,\\
a^3\cos a \sin b &= b^3\sin a \cos b.
\end{align*} 
First, we show that the points \eqref{SP:intersections} satisfy the two conditions. In the $(a,b)$ plane, these points correspond to the following two families of points:
\begin{align*}
C^{\mathrm{odd}}_{j,k}(a, b) &= \left(\left(j+\frac{1}{2}\right)\pi, \left(k+\frac{1}{2}\right)\pi\right),\\
C^{\mathrm{even}}_{j, k}(a, b) &= \left(j\pi, k\pi\right),
\end{align*} 
where $j-k=l$. It is obvious that two pairs satisfy both eigenvalue conditions simultaneously. 

Next we show that these are the only solutions of the two conditions. By adding the two equations and using a trigonometric identity, we obtain
\begin{align*}
a^3\sin(a+b) &= b^3\sin(a+b).
\end{align*}
Thus we must have either $a=b$ or $a+b = m\pi, m\in \mathbb{N}$. However, in this case we know that $a\neq b$, since we are not on the critical parabola (see the proof of Lemma~\ref{lemma:eigenfunctionsLHP}). Similarly, by subtracting eigenvalue conditions, we obtain the condition $a-b=n\pi, n\in\mathbb{N}$. 

Additionally, we know from the odd condition that $a\in\left(\left(j-1/2\right)\pi, \left(j+1/2\right)\pi\right]$ and $b\in\left(\left(k-1/2\right)\pi, \left(k+1/2\right)\pi\right]$ for $j>k$. Hence $a-b\in\left(\left(j - k -1\right)\pi, \left(j - k +1\right)\pi\right)$, or equivalently $\left(\left(l-1\right)\pi, \left(l+1\right)\pi\right)$. Similarly, we obtain $a-b\in
=\left(\left(l-1\right)\pi, \left(l+1\right)\pi\right)$ for the even condition. Hence $a-b=l\pi$. Therefore, $a = (m+l)\pi/2$ and $b = (m-l)\pi/2$. 

We have thus identified all points where the odd and even eigenvalue curves meet. We then recover the $(\tau,\mu)$ values of the points of intersection from the bijection $F_2$.
\end{proof}

We now list, as observations, some properties of the eigenvalue branches. 
\begin{enumerate}
\item \label{prop:LHPconnection}\textbf{Observed connection between the upper half-plane and super-parabolic regions:} In the upper half-plane, eigenvalue branches (save for the constant zero branch) are indexed by nonnegative integers $l$; for each $l$, there is one odd branch and one even branch. In the super-parabolic region, the parameterization of eigenvalue branches depends upon indices $j$ and $k$. If we have $l=j-k$, then then the $l$th odd eigenvalue branch in the upper half-plane and the $j,k$th odd eigenvalue branch in the super-parabolic region have the same horizontal intercept, $\tau=-l^2\pi^2$. The same is true for even branches.

\item\textbf{Asymptotic behavior of the eigenvalue branches.}

Along each branch, as $\tau$ tends to $-\infty$, we have
\[
\mu = -\frac{\tau^2}{4} + \mathcal{O}(|\tau|),
\] 
where the constant in the ``$\mathcal{O}$" term depends on the eigenvalue branch.
\begin{proof}(Sketch)

From Property~(\ref{prop:LHPconnection}), we may write $a=b+l\pi+(s-t)$, where $s, t\in(-\pi/2, \pi/2]$. Set $\gamma=l\pi+(s-t)$. Note that while $\gamma$ is variable, it takes values in $\left(\left(l-1\right)\pi, \left(l+1\right)\pi\right)$ and so is bounded. We express $\tau$ and $\mu$ in terms of $b$ and $\gamma$ as follows:
\begin{align*}
\mu &= -a^2b^2=-b^4-2\gamma b^3-\gamma^2b^2,\\
\tau &= -a^2-b^2=-2b^2-2\gamma b-\gamma^2.
\end{align*}
From this, we obtain a relation between $\tau$ and $\mu$:
\[
\mu = -\frac{\tau^2}{4} + \frac{(-\gamma^2\pm 2\gamma\sqrt{-\gamma^2-2\tau})^2}{4}
\]
Therefore, $\mu=-\tau^2/4-2\gamma^2\tau+\mathcal{O}\left(|\tau|^{1/2}\right)$, where $\gamma$ depends on the branch $l$ as well as $s$ and $t$. 
\end{proof}

\item\label{prop:LHPphantom} \textbf{Observation of shallow straight lines (Cascading phenomenon)}
We observe that the eigenvalue curves in the super-parabolic region consist of a pattern of nearly-linear segments and barely-avoided crossings. If we draw lines through the even-odd points of intersection as in Figure~\ref{fig:LHPphantom}, we see they are very close to the nearly-linear segments of different branches of the same symmetry type.  We will call these lines `phantom spectral lines'.

From Proposition~\ref{prop:LHPintersections}, we know that the two families of intersection points indexed by the same $l$ can be expressed as
\begin{align*}
C_{j, k}^{\mathrm{odd}}(\tau, \mu) &= \left(\frac{-\left(2j+1\right)^2-\left(2k+1\right)^2}{4}\pi^2, \frac{-\left(2j+1\right)^2\left(2k+1\right)^2}{16}\pi^4\right),\\
C_{j, k}^{\mathrm{even}}(\tau, \mu) &= \left(\left(-j^2-k^2\right)\pi^2, -j^2k^2\pi^4\right) \quad\text{where $j-k=l$.}
\end{align*}
The $k$th phantom spectral line for the odd branches is the line that connects the points $C^{odd}_{j,k}$ with $j>k$. Similarly, the $k$th phantom spectral line for the even branch connects the points $C^{even}_{j,k}$ with $j\geq k$. The equations of these phantom spectral lines are:
\begin{align*}
P_{k}^{\mathrm{odd}}: \mu &= \left(\frac{(2k+1)\pi}{2}\right)^2\tau +\left(\frac{(2k+1)\pi}{2}\right)^4,\quad\text{for $k\geq 0$},\\
P_{k}^{\mathrm{even}}: \mu &= (k\pi)^2\tau + (k\pi)^4,\quad\text{for $k\geq 0$}.
\end{align*}

  \begin{figure}[H]
    \begin{center}
\includegraphics[width=2.5in]{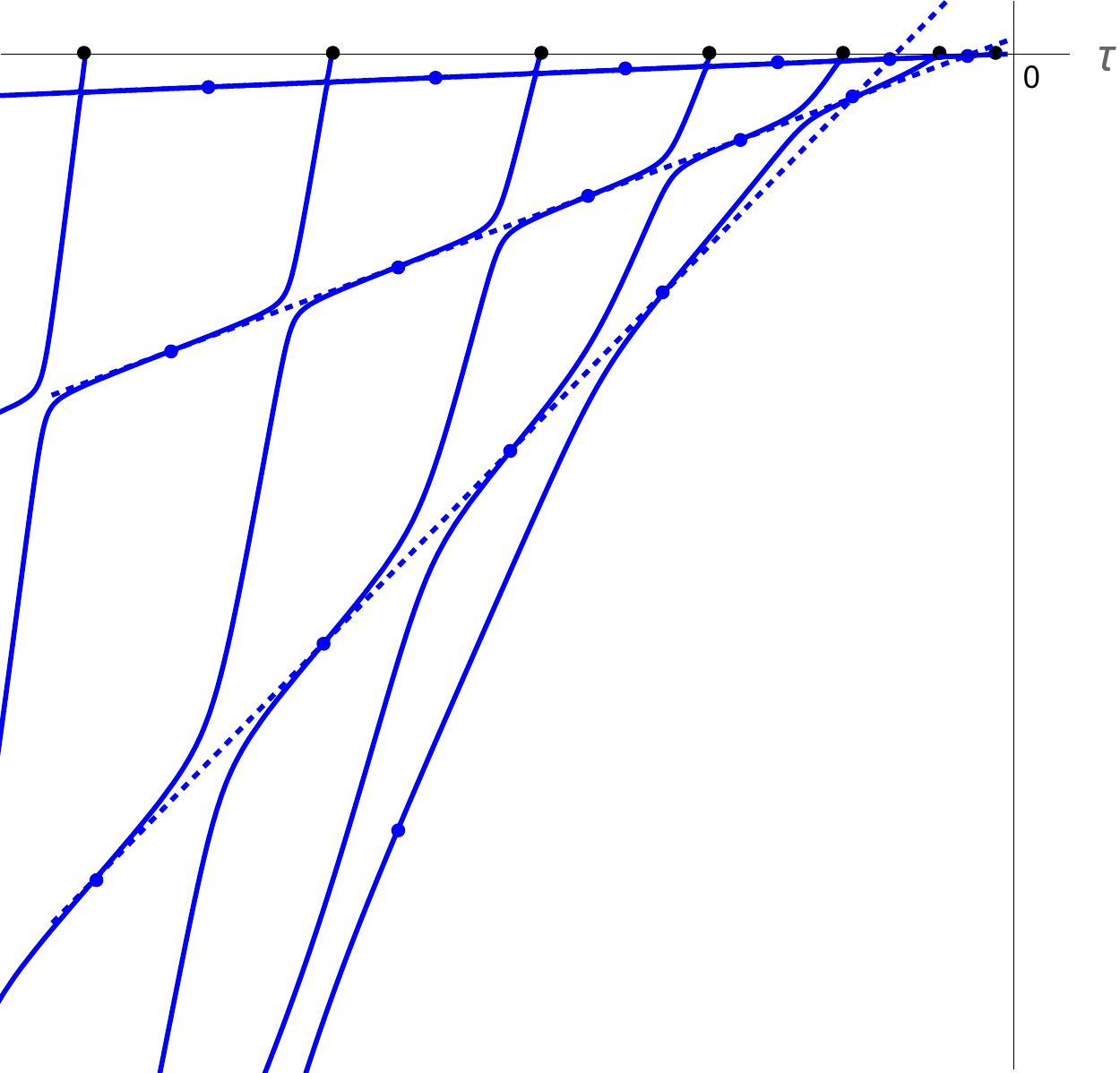}
\qquad
\includegraphics[width=2.5in]{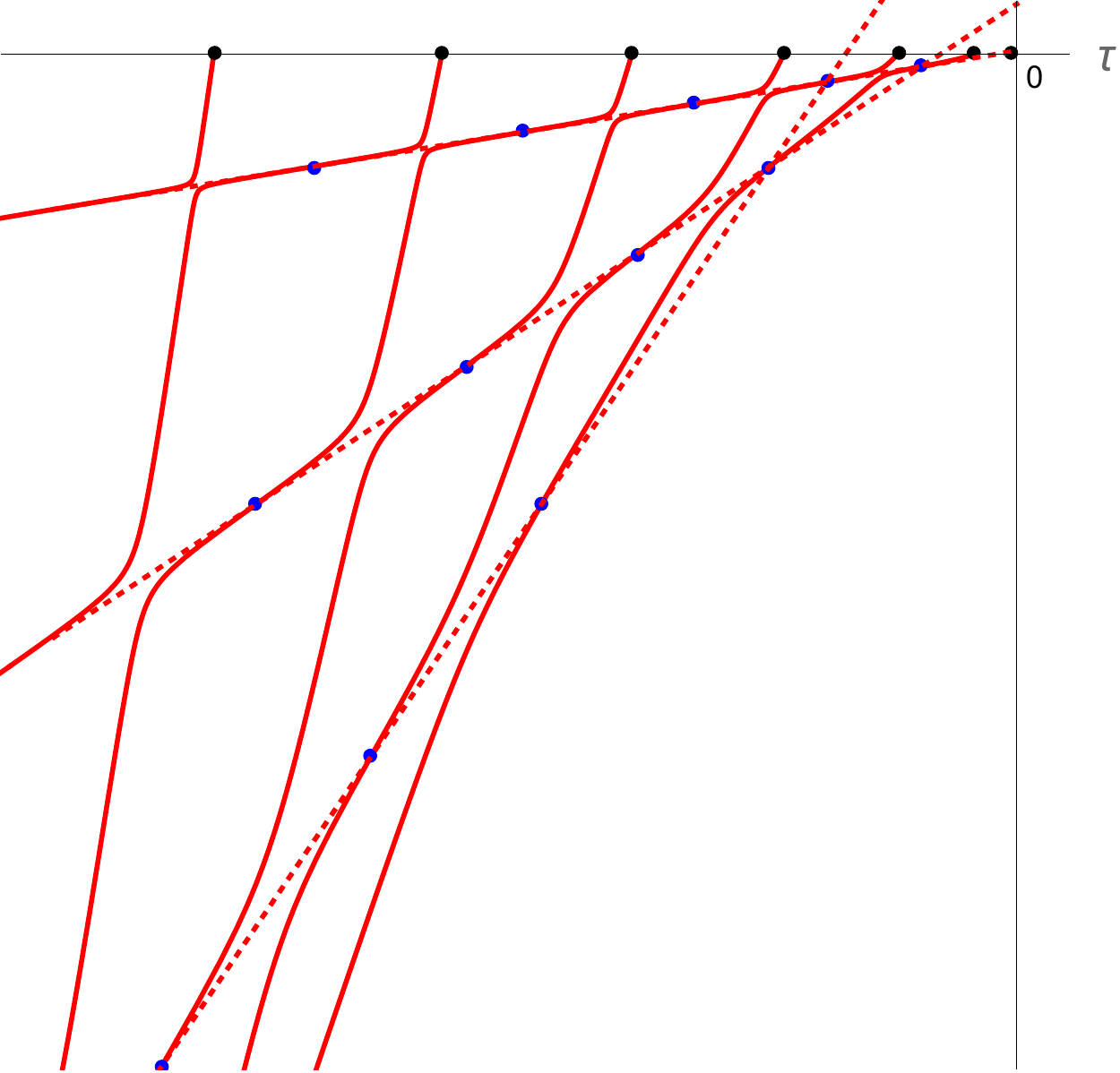}
\caption{\label{fig:LHPphantom} Lower half-plane above the critical parabola: dotted ``phantom" lines for the odd and even eigenvalue branches (Property~\ref{prop:LHPphantom}).}
    \end{center}
    \end{figure}

\begin{remark}
Figure ~\ref{fig:LHPphantom} suggests that the phantom spectral lines are tangent to the eigenvalue curves at each points $C_{j, k}^{\mathrm{odd}}$'s and $C_{j, k}^{\mathrm{even}}$'s. This is not the case, although as $\mu\to-\infty$, the slopes of the phantom spectral lines approach those of the eigenvalue curves. The details will be in the second author's thesis, to be published later.
\end{remark}

\end{enumerate}

\section{\bf The lower half-plane: sub-parabolic region}\label{sec:crit}
In this section, which is the most technically difficult of the paper, we treat the case of eigenvalue branches lying in the lower half-plane below the critical parabola, that is, $\{(\tau,\mu) : \tau < 0,\mu < -\tau^2 /4\}$. We will see that this region corresponds to the case where the characteristic equation $r^4-\tau r^2 -\mu=0$ has non-real, non-purely-imaginary complex roots. Note that since a nonnegative $\tau$ guarantees nonnegative eigenvalues, we need only consider $\tau<0$. When $\mu$ and $\tau$ are both negative and satisfy $\mu < -\tau^2 /4$, we may factor the eigenvalue equation as
\begin{equation}\label{eqn:factorcrit}
\left(\frac{d^2}{dx^2}+(a-ib)^2\right)\left(\frac{d^2}{dx^2}+(a+ib)^2\right)u=0,
\end{equation}
for $a>b>0$, where $\mu = -(a^2+b^2)^2$ and $\tau = 2(b^2-a^2)$, as in the bijection Lemma~\ref{lemma:Bij}(\ref{lemma:crit}).

The behavior of the eigenvalue branches is quite different in this region from that in the upper half-plane or the super-parabolic region. We cannot find an explicit parameterization, so instead we develop an ``implicit parameterization" to describe the odd and even eigenvalue branches. In particular, we will show that there are only two eigenvalue branches (one each of odd and even), and they cross infinitely many times. The eigenfunctions have a more complicated form than in previous cases. 

We will also discuss in Section~\ref{sec:intersectionsEP} the question of intersections of the family of parabolas $\mu=-c\tau^2$ with the first two eigenvalue branches, which has applications to the the second author's thesis (to be published later).

\begin{figure}[H]
    \begin{center}
\includegraphics[scale=0.6]{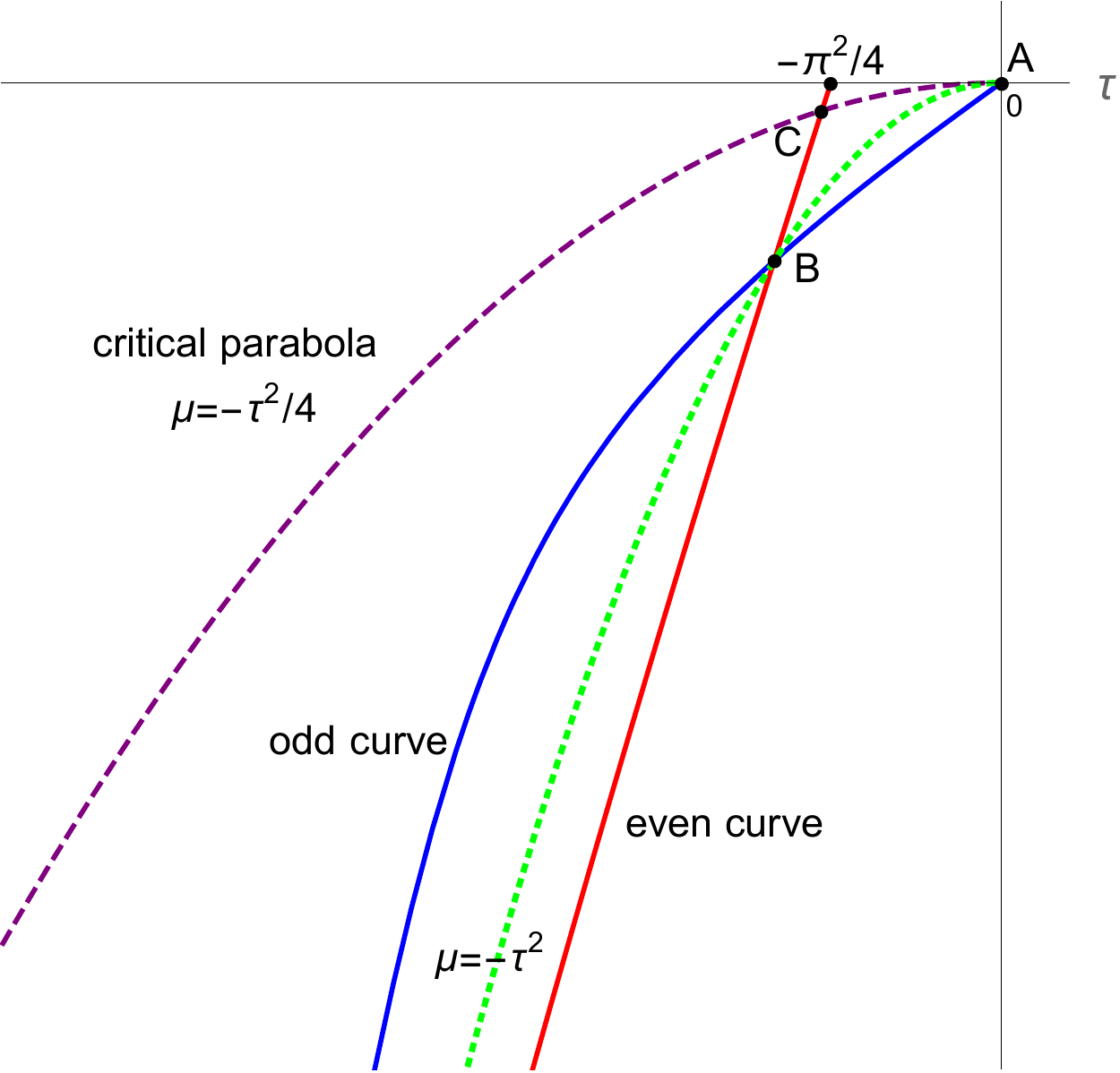}
\caption{\label{fig:crit} Only two spectral curves penetrate below the critical parabola. They intersect infinitely often along the dotted parabola $\mu=-\tau^2$. The points $A$, $B$, $C$, and $(-\pi^2/4, 0)$ are significant for discussion in Section~\ref{sec:crit}.}     \end{center}
    \end{figure}

\subsection{\bf Eigenfunctions and Eigenvalue conditions}

\begin{lemma}[Sub-parabolic region]  \label{lemma:eigenfunctionscrit}
For all $\tau<0$ and all eigenvalues $\mu$ satisfying $\mu<-\tau^2 /4$, at least one of the following must hold: 
\begin{enumerate}
\item The eigenvalue $\mu$ is associated with an odd eigenfunction $u_o$ of the form
\[
u_o(x) = A \cos(ax) \sinh(bx) + B \sin(ax) \cosh(bx),
\]
where $A$ and $B$ are nonzero constants, and $a>b>0$ are positive numbers such that $\mu = -(a^2 + b^2)^2$, $\tau = 2(b^2-a^2)$, and satisfy the additional condition
\begin{equation}\label{eqn:critevalCondOdd}
(3a^2 - b^2)\frac{\sin (2a)}{2a} = (3b^2 - a^2)\frac{\sinh(2b)}{2b}.
\end{equation}
\item The eigenvalue $\mu$ is associated with an even eigenfunction $u_e$ of the form 
\[
u_e(x) = C \sin(ax) \sinh(bx) + D \cos(ax)\cosh(bx),
\]
where $C$ and $D$ are nonzero constants, and $a>b>0$ are positive numbers such that $\mu = -(a^2 + b^2)^2$, $\tau = 2(b^2-a^2)$, and satisfy the additional condition
\begin{equation}\label{eqn:critevalCondEven}
(3a^2 - b^2)\frac{\sin (2a)}{2a} = (a^2 - 3b^2)\frac{\sinh(2b)}{2b}.
\end{equation}
\end{enumerate}
  \begin{figure}[t]
    \begin{center}
\includegraphics[scale=0.55]{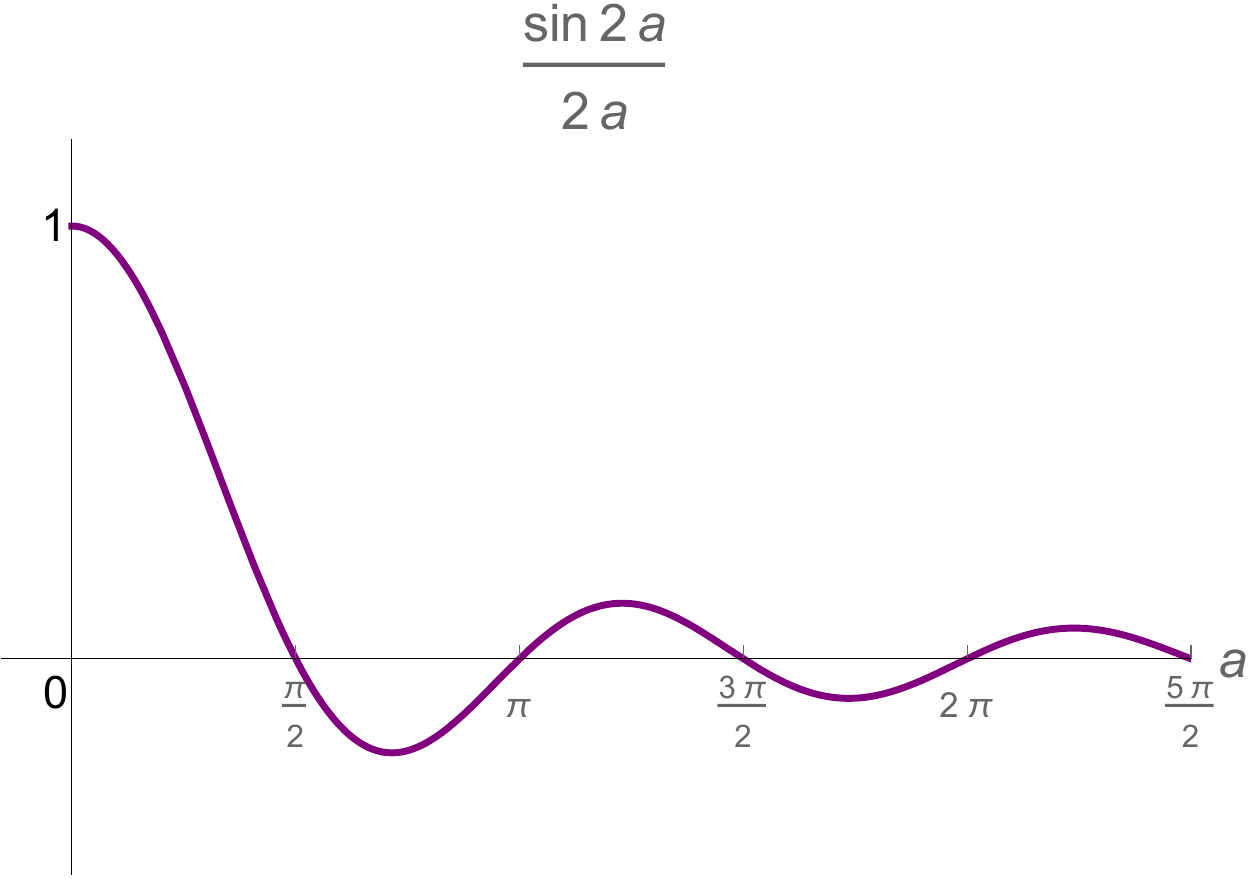}
\qquad
\includegraphics[scale=0.55]{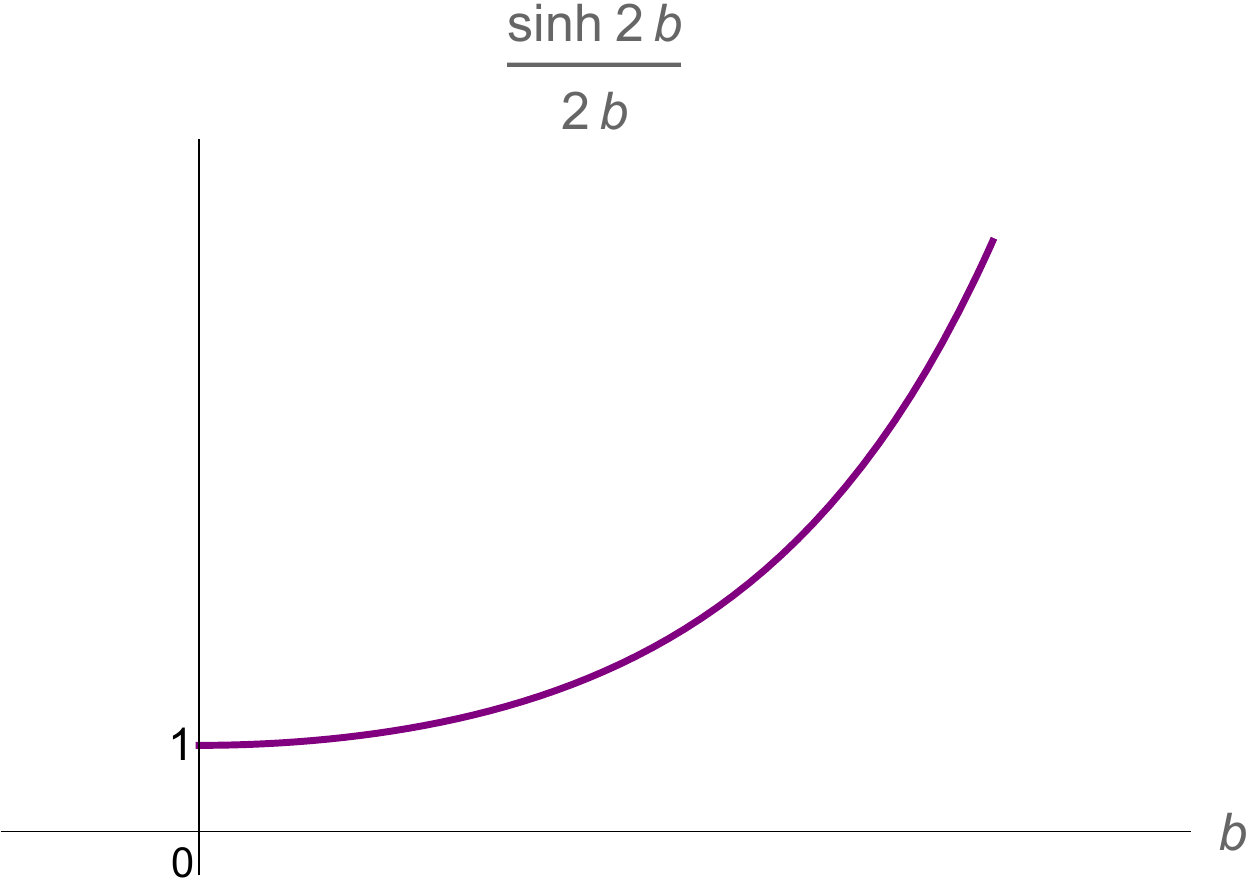}
\caption{\label{fig:critEvalcondition}
Sub-parabolic region: graphs involved in eigenvalue conditions \eqref{eqn:critevalCondOdd}-\eqref{eqn:critevalCondEven} in Lemma~\ref{lemma:eigenfunctionscrit}.}
    \end{center}
    \end{figure}
\end{lemma}

\begin{proof} Since we are considering the case of $\mu<-\tau^2/4$ and $\tau<0$, by Lemma~\ref{lemma:Bij}(\ref{lemma:crit}), we may express the eigenvalue equation in terms of $a$ and $b$ as \eqref{eqn:factorcrit}, and so the characteristic equation has four distinct non-real, non-purely-imaginary complex roots $r=\pm ia \pm b$. As usual, we consider only even and odd solutions, so we express the solutions $e^{\pm iax \pm bx}$ of the differential equation as $\cos(ax) \sinh(bx)$, $\sin(ax) \sinh(bx)$, $\cos(ax) \cosh(bx)$, and $\sin(ax) \cosh(bx)$. Our possible solutions are then linear combinations of these, chosen according to symmetry. We then see the odd and even eigenfunctions have the forms
\begin{align*}
u_o(x) &= A \cos(ax) \sinh(bx) + B \sin(ax) \cosh(bx),\\
u_e(x) &= C \sin(ax) \sinh(bx) + D \cos(ax)\cosh(bx).
\end{align*}
The boundary conditions \eqref{eqn:1dbcm} and \eqref{eqn:1dbcv} in terms of $a$ and $b$ say that:
\[
\begin{cases}
u'' = 0 &\text{when $x=\pm 1$,}\\
u''' - 2(b^2-a^2)u' = 0 &\text{when $x = \pm 1$.}
\end{cases}
\] 

We consider the odd eigenfunction first. We wish to determine which values of $a$ and $b$ (and hence $\tau$ and $\mu$) yield a solution to the boundary value problem. Applying the two boundary conditions yields
\[
\begin{cases}
&A[(b^2-a^2)\cos a \sinh b - 2ab\sin a\cosh b] \\
 &\qquad= -B[(b^2 - a^2)\sin a \cosh b + 2ab \cos a\sinh b],\\
&A[(b^3 + a^2b)\cos a \cosh b + (a^3 + ab^2) \sin a \sinh b] \\
&\qquad= -B[(b^3 + a^2b)\sin a\sinh b -(a^3+ab^2)\cos a\cosh b].
\end{cases}
\] 
We require that our linear combination coefficients $A,B$ be nontrivial, so the system's determinant must vanish. We can express this determinant condition as
\[
(4a^2b-2b^3+2a^2b)\sin a \cos a = (2ab^2 - 2a^3 + 4ab^2)\sinh b\cosh b.
\]
Since $a, b$ are nonzero, this formula is equivalent to
\[
(3a^2 - b^2)\frac{\sin (2a)}{2a} = (3b^2 - a^2)\frac{\sinh(2b)}{2b},
\]
which gives a condition on $(a,b)$ that guarantees existence of an odd solution to the eigenvalue problem.

For the even eigenfunction $u_e$, applying the boundary conditions gives us
\[
\begin{cases}
&C[(b^2-a^2)\sin a \sinh b + 2ab\cos a\cosh b] \\
&\qquad= -D[(b^2 - a^2)\cos a \cosh b -2ab \sin a\sinh b],\\
&C[(a^2b + b^3) \sin a \cosh b - (a^3 + ab^2)\cos a \sinh b] \\
&\qquad= -D[(a^2b+b^3)\cos a\sinh b + (a^3 + ab^2)\sin a\cosh b].
\end{cases}
\]
Once again, to have nontrivial linear combination coefficients $C,D$, we require the system's determinant be zero, which leads to
\[
(3a^2 - b^2)\frac{\sin (2a)}{2a} = (a^2 - 3b^2)\frac{\sinh(2b)}{2b}.\qedhere
\]
\end{proof}

Unlike the prior cases, we cannot solve the eigenvalue conditions \eqref{eqn:critevalCondOdd} or \eqref{eqn:critevalCondEven} for $a$ in terms of $b$ (or $b$ in terms of $a$) explicitly, and so we have not found a smooth parameterization of the eigenvalue curves in this region. Instead, we seek to understand the eigenvalue branches by looking at solutions of the eigenvalue conditions in the $(a, b)$-plane. We call this method ``implicit parameterization'' since it relies on implicit functions and does not give us a complete parameterization as we found for the other regions. 

In the sections that follow, we describe the behavior of the eigenvalue curves in the $(a, b)$-plane. We then consider the problem of intersections of parabolas and the eigenvalue curves by considering their images in the $(a,b)$-plane.

\subsection{\bf Image of eigenvalue branches in the $(a, b)$-plane}
In this section, we consider the images of the eigenvalue branches of the sub-parabolic region in the $(a, b)$-plane. From Lemma~\ref{lemma:eigenfunctionscrit}, we know that the odd and even eigenvalue curves in the $(a, b)$-plane are given by the following equations:
\begin{align*}
\text{ odd}:\quad (3a^2 - b^2)\frac{\sin (2a)}{2a} &= (3b^2 - a^2)\frac{\sinh(2b)}{2b} ,\\
\text{ even}:\quad (3a^2 - b^2)\frac{\sin (2a)}{2a} &= -(3b^2 - a^2)\frac{\sinh(2b)}{2b} .
\end{align*}
The shapes of the eigenvalue curves are not clear from the equations themselves. We analyze their properties near the origin in the next two lemmas.

Define the functions
\begin{align*}
F_o(a, b) = (3a^2 - b^2)\frac{\sin (2a)}{2a} - (3b^2 - a^2)\frac{\sinh(2b)}{2b},\\
F_e(a, b) = (3a^2 - b^2)\frac{\sin(2a)}{2a} + (3b^2 - a^2)\frac{\sinh(2b)}{2b}.
\end{align*}
We restrict our domain to be the triangle 
\[
\{(a, b) : 0\leq a\leq \pi/2, 0\leq b\leq \pi/\sqrt{12}, b<a\},
\]
continuously extending $\sin(2a)/2a$ and $\sinh(2b)/2b$ to $a=0$ and $b=0$ respectively. 

First, we will show that $F_o(a, b)=0$ has a unique solution $b(a)$ for each $a\in(0, \pi/2)$.  We then use the Implicit Function Theorem to conclude that these solutions form a continuous function of $a$, and so the graph is a single connected curve in the $(a,b)$-plane. We argue the analogous result is also true for the even case: that $F_e(a, b)=0$ has a unique solution $a(b)$ for each $b\in(0, \pi/\sqrt{12})$, which can be considered a continuous function.

We first prove existence and uniqueness of the solutions.
 
\begin{lemma}[Existence and uniqueness of the solutions]\label{lemma:euoddeven}
\begin{enumerate}
 \item \label{lemma:euoddevenOdd}
 For each $a\in (0, \pi/2)$, there exists a unique $b\in(a/\sqrt{3}, a)$ such that $F_o(a, b)=0$, and furthermore, the equation has no solutions when $b\in(0,a/\sqrt{3}]$. 
 
 \item 
 For each $b\in (0, \pi/\sqrt{12})$, there exists a unique $a\in(\sqrt{3}b, \pi/2)$ such that $F_e(a, b)=0$, and furthermore, the equation has no solutions when $a\in(b, \sqrt{3}b]$.
\end{enumerate}
\end{lemma}

\begin{proof} Let $f(a)=\sin(2a)/2a$ and $g(b)=\sinh(2b)/2b$. Note that $f(a)>0$ for $a\in(0,\pi/2)$ and $f(a) <1<g(a)$ for all $a>0$; see Figure~\ref{fig:critEvalcondition}.

We begin with the odd case. We fix $a\in(0,\pi/2)$ and show existence of $b$ satisfying $F_o(a,b)=0$. 

For any such $a$, it is obvious that $F_{o}(a,\cdot)$ is a continuous function on $b\in[a/\sqrt{3}, a]$. At the left endpoint $b=a/\sqrt{3}$, we have $F_{o}(a, a/\sqrt{3}) = 8a^2f(a)/3 > 0$, since $a\in(0,\pi/2)$. At the right endpoint $b=a$, we have $F_{o}(a, a) = 2a^2(f(a) - g(a)) <0$, since $a>0$. Then by the Intermediate Value Theorem, there exists a $b\in(a/\sqrt{3}, a)$ such that $F_{o}(a, b) = 0$. 

To show uniqueness of the $b$ for the odd case, we fix the value $a$ and write the equation $F_{o}(a, b) = 0$ in the form
\begin{align*}
g(b) = f(a)h(b), \qquad \text{where $h(b)= \frac{3a^2-b^2}{3b^2-a^2}$,}
\end{align*}
and consider the behavior of both sides. The left-hand side $g(b)$ is positive and increasing in $b$. On the other hand, $f(a)$ is a fixed positive number and $h(b)$ is positive and decreasing when $b\in(a/\sqrt{3}, a)$. Therefore, there is only one $b\in(a/\sqrt{3}, a)$ that satisfies the equation. 

Note also that if $0<b\leq a/\sqrt{3}$, there is no solution for $F_{o}(a, b) = 0$. On this interval, the function $h(b)$ is negative, but $g(b)$ remains positive.

The even case proceeds similarly. In this case, we fix $b\in(0, \pi/\sqrt{12})$ and consider $F_{e}(\cdot,b)$ as the continuous function on $a\in[\sqrt{3}b, \pi/2]$. It is easy to show $F_{e}(\sqrt{3}b, b)>0$ and $F_{e}(\pi/2,b)<0$.

To show uniqueness of that $a$-value for the even case, we fix the value $b$ and write the equation $F_{e}(a, b) = 0$ in the form
\begin{align*}
f(a) = g(b)h(a), \qquad \text{where $h(a)=-\frac{3b^2-a^2}{3a^2-b^2}$,}
\end{align*}
and consider behavior of both sides. The left hand side $f(a)$ is positive and decreasing. On the other hand, $g(b)$ is a fixed positive number and $h(a)$ is positive and increasing when $a\in(\sqrt{3}b, \pi/2)$. Therefore, there is only one $a\in(\sqrt{3}b, \pi/2)$ that satisfies the equation. 

Similarly, if $b<a\leq \sqrt{3}b$, there is no solution for $F_{e}(a, b) = 0$. On this interval, the quantity $h(a)$ is now negative, but $f(a)$ remains positive.
\end{proof}

\begin{remark} This method of proof can be extended to show existence and uniqueness of $b(a)$ for all $a>0$ in the odd case, with $b\in(0,a/\sqrt{3}]$ whenever $a\in(k\pi,(2k+1)\pi/2)$ and $b\in [a/\sqrt{3},a]$ when $a\in[(2k+1)\pi/2,(k+1)\pi]$, where $k$ is an integer. Existence of $a(b)$ for all $b\geq 0$ the even case can likewise be obtained, but establishing uniqueness requires another method. 
\end{remark}

We next show that eigenvalue conditions $F_o(a, b)=0$ and $F_e(a, b)=0$ actually have continuous solutions $b(a)$ and $a(b)$, respectively. For what follows, we find it useful to define the quantity $a^\ast$ such that $\sin(2a^\ast)/2a^\ast=1/3$, which corresponds to $F_e(a^\ast, 0)=0$. Numerically, we have 
\begin{equation} \label{eqn:aast}
a^\ast\approx 1.13943.  
\end{equation}

\begin{lemma}[Continuity and smoothness of the solutions]\label{lemma:smoothness}
\begin{enumerate}
 \item The solution $b(a)$ of $F_o(a, b)=0$ is continuous for $a\in[0, \pi/2]$, smooth for $a\in(0, \pi/2)$, and satisfies $b(0)=0$, $b'(0^{+})=1$, and $b(a)<a$ for $a\in(0, \pi/2)$.
 \item The solution $a(b)$ of $F_e(a, b)=0$ is continuous for $b\in[0, \pi/\sqrt{12}]$ and smooth for $b\in(0, \pi/\sqrt{12})$.
\end{enumerate}

\end{lemma}
\begin{proof} By Lemma~\ref{lemma:euoddeven}, there exist unique solutions of $F_o(a, b)=0$ for each $a\in(0, \pi/2)$ and of $F_e(a, b)=0$ for $b\in (0, \pi/\sqrt{12})$. We will use the Implicit Function Theorem to establish these solutions depend continuously on $a$ or $b$ as appropriate. In order to prove smoothness and properties of the solution, we break into two pieces: ``away from the origin" and ``near the origin".

\textbf{Claim 1:} The solution $b(a)$ of $F_o(a, b)=0$ is smooth for $a\in(0, \pi/2)$ (away from the origin) and $b(a)<a$ for all $a\in(0, \pi/2)$.

From Lemma~\ref{lemma:euoddeven}(\ref{lemma:euoddevenOdd}), it is obvious that the solution $b(a)<a$ for $a\in(0, \pi/2)$. We will use the Implicit Function Theorem to show the solution is smooth for $a\in(0, \pi/2)$.
\begin{enumerate}
\item First, we show the function $F_o$ is continuously differentiable for $(a, b)\in [0, \pi/2] \times [0, \pi/\sqrt{12}]$.

Note that $F_o$ is analytic as a function of $(a, b)\in\RR^2$, by using the standard power series expansions for $\sin(z)$ and $\sinh(z)$. That is, the function $F_o$ can be written as
\[
F_o(a, b) = (3a^2-b^2)\left(1-\mathcal{O}(a^2)\right) - (3b^2-a^2)\left(1+\mathcal{O}(b^2)\right).
\] 
Since $F_o$ is an analytic function on $(a, b)\in\RR^2$, it is easy to see that the partials $\partial F_o/\partial a$ and $\partial F_o/\partial b$ are continuous for $(a, b)\in [0, \pi/2] \times [0, \pi/\sqrt{12}]$.

\item Next, we show $\partial F_o/\partial b$ is nonzero for $(a, b)\in (0, \pi/2] \times (0, \pi/\sqrt{12}]$. Recall from \eqref{eqn:critevalCondOdd} that for the odd eigenvalue branch $F_o(a, b)=0$, we have
\begin{align*}
\frac{\partial F_o}{\partial b} = -\frac{b}{a}\sin(2a) -\left(\frac{a^2}{2b^2}+\frac{3}{2}\right)\sinh(2b) + \left(\frac{a^2}{b}-3b\right)\cosh(2b).
\end{align*} 
The first and second terms are obviously negative for $a\in(0, \pi/2]$ and $b\in(0, \pi/\sqrt{12}]$. The third term is also negative, since by Lemma~\ref{lemma:euoddeven} the solution curve $F_o(a, b)=0$ lies in the portion of the plane where $a<\sqrt{3} b$. Therefore, $\partial F_o/\partial b<0$.  
\end{enumerate}

Therefore, by the Implicit Function Theorem, there exist neighborhoods $U_{a}\subseteq\RR$  and $V_{b}\subseteq\RR$ such that $a\in U_{a}, b\in V_{b}$, and a unique function $g_{a, b}: U_a \to V_b$ exists such that $b=g_{a, b}$ and $F_o(x, g_{a, b}(x))=0$ for all $x\in U_{a}$. Furthermore, the function $g_{a, b}$ is infinitely differentiable on $U_a$.

We cannot apply the Implicit Function Theorem at the origin because $\partial F_o/\partial b=0$ there. To understand the behavior of the curve near the origin, we introduce a change of variable.

\textbf{Claim 2:}  The solution $b(a)$ of $F_o(a, b)=0$ is smooth on a neighborhood of the origin, and $b(0)=0$, $b'(0^{+})=1$, and $b(a)<a$ for $a$ near $0$.

We express $F_o(a,b)$ near the origin using the standard power series expansions for $\sin(2a)$ and $\sinh(2b)$, obtaining
\begin{align*}
F_o &= (3a^2-b^2)\left(1-\frac{2}{3}a^2 + \mathcal{O}(a^4)\right) - (3b^2-a^2)\left(1+ \frac{2}{3}b^2 + \mathcal{O}(b^4)\right).
\end{align*}
Expressing $F_o$ as a function of the new variables $\alpha=a^2$ and $\beta=b^2$, we obtain a series expansion:
\begin{align*}
F_o = (3\alpha-\beta)\left(1-\frac{2}{3}\alpha + \mathcal{O}(\alpha^2)\right) - (3\beta-\alpha)\left(1+ \frac{2}{3}\beta + \mathcal{O}(\beta^2)\right)
\end{align*}
Note that $\partial F_o/\partial \beta=-4\neq0$ at $(\alpha, \beta)=(0, 0)$. Hence by the Implicit Function Theorem applied to $F_o=0$ in the $(\alpha, \beta)$-plane, we obtain a unique smooth solution $\beta=\beta(\alpha)$. The condition $F_o=0$ implies that
\begin{align*}
4(\alpha-\beta)=\frac{2}{3}(\alpha+\beta)^2+\frac{4}{3}(\alpha-\beta)^2+\mathcal{O}(\rho^3),
\end{align*} 
where $\rho=\sqrt{\alpha^2+\beta^2}$. The right-hand side of the above equation is positive near $(\alpha, \beta)=(0, 0)$, so we conclude $\alpha>\beta$ near the origin on the curve defined by $F_o=0$.  Moreover, we know that at $(\alpha, \beta)=(0, 0)$,
\begin{align*}
\frac{\partial \beta}{\partial \alpha} = -\frac{\partial F_o/\partial \alpha}{\partial F_o/\partial \beta}=1.
\end{align*}
Therefore, we have a smooth curve $\beta(\alpha)=\alpha+\mathcal{O}(\alpha^2)$ that passes through the origin in the $(\alpha, \beta)$-plane with slope $1$, and $\alpha>\beta$. By reverting to the original variables, we obtain $b^2=a^2+\mathcal{O}(a^4)$. That is, $b=a+\mathcal{O}(a^3)$ and so $b(0)=0$, $b'(0^{+})=1$, and $b(a)<a$ for $a$ near the origin. 
 
We now repeat this argument for the even branch. We consider the equation $F_e(a, b)=0$ for all $a\in[a^\ast, \pi/2]$ and $b\in[0, \pi/\sqrt{12}]$, where $a^\ast$ is as defined previously (see \eqref{eqn:aast}). We need $\partial F_e/\partial a$ to be nonzero at each point $(a, b)$ along the solution curve.

Recall that for the even eigenvalue branch $F_e(a, b)=0$, we have
\[
\frac{\partial F_e}{\partial a} = \left(\frac{3}{2}+\frac{b^2}{2a^2}\right)\sin(2a)+\left(3a-\frac{b^2}{a}\right)\cos(2a)-\frac{a}{b}\sinh(2b).
\]
For convenience, we divide $\partial F_e/\partial a$ by $a>0$ and show that is always a negative quantity for the restricted range. After rewriting this expression, we have
\begin{equation}\label{eqn:dFeda}
\frac{1}{a}\frac{\partial F_e}{\partial a} = \left(3+\frac{b^2}{a^2}\right)\frac{\sin(2a)}{2a}+\left(3-\frac{b^2}{a^2}\right)\cos(2a)-\frac{\sinh(2b)}{b}.
\end{equation}
Solving $F_e(a, b)=0$ for the expression $\sin(2a)/2a$, we obtain
\begin{align*}
\frac{\sin(2a)}{2a} = -\left(\frac{3b^2-a^2}{3a^2-b^2}\right)\frac{\sinh(2b)}{2b}.
\end{align*} 
We substitute this expression into \eqref{eqn:dFeda}, obtaining
\[
\frac{1}{a}\frac{\partial F_e}{\partial a} = -\left(\frac{3(a^2+b^2)^2}{a^2(3a^2-b^2)}\right)\frac{\sinh(2b)}{2b} + \left(\frac{3a^2-b^2}{a^2}\right)\cos(2a).
\]
The first term is negative for $a\in[a^\ast, \pi/2]$ and $b\in[0, \pi/\sqrt{12}]$, since $\sinh(2b)/2b>0$ for all $b\in[0, \pi/\sqrt{12}]$ and the curve $F_e(a, b)=0$ lies in the region $a>\sqrt{3}b$ by Lemma~\ref{lemma:euoddeven}. 
The second term is also negative since we consider $a\in[a^\ast, \pi/2]$ and $a^\ast\approx 1.13943>\pi/4$. Therefore, we obtain that $\partial F_e/\partial a$ is negative (and so nonzero) at each point $(a, b)$ with $a\in[a^\ast, \pi/2]$ and $b\in[0, \pi/\sqrt{12}]$. 

By the Implicit Function Theorem, there exist neighborhoods $U_{a}\subseteq\RR$ and $V_{b}\subseteq\RR$ such that $a\in U_{a}, b\in V_{b}$, and a unique function $g_{a, b}: V_b \to U_a$ exists such that $a=g_{a, b}(b)$ and $F_e(g_{a, b}(y), y)=0$ for all $y\in V_{b}$. Furthermore, the function $g_{a, b}$ is infinitely differentiable on $V_{b}$.
\end{proof} 

Now that we have achieved smooth curves  $F_o(a, b)=0$ and $F_e(a, b)=0$ in the $(a,b)$ plane, the bijection $F_3$ gives us a homeomorphism of these curves onto smooth curves in the $(\tau, \mu)$-plane. These are the odd eigenvalue curve connecting the points $A(0, 0)$ and $B(-\pi^2/3, -\pi^4/9)$, and the even eigenvalue curve connecting the points $B(-\pi^2/3, -\pi^4/9)$ and $C(-2(a^\ast)^2, -(a^\ast)^4)$, the latter of which lies on the critical parabola (see Figure~\ref{fig:crit}). 

In the sections which follow, we will use the notation $\mu_o(\tau)$ and $\mu_e(\tau)$ to denote the odd and even eigenvalue branches in the sub-parabolic region.
 
\subsection{\bf Intersections of a family of parabolas and the lowest eigenvalues $\mu_o(\tau)$ and $\ 
\mu_e(\tau)$} \label{sec:intersectionsEP} 

Our goal is now to prove the following proposition:
\begin{proposition}[Intersection between parabola and the eigenvalues $\mu_o(\tau)$ or $\mu_e(\tau)$]\label{prop:intersectionsEP}
Each parabola $\mu = - c\tau^2$ with $c>0$ intersects at least one eigenvalue branch in the third quadrant.
\end{proposition}

In other words, we wish to demonstrate that for each $c>0$, there exists a point $(\tau, \mu)$ on the parabola $\mu=-c\tau^2$ which is an eigenvalue pair. We build up the proof in several steps.

When $0<c\leq 1/4$, the parabola lies on or above the critical parabola, and we will be able to work directly with our parameterizations from the super-parabolic region. When $c>1/4$, this parabola lies in the sub-parabolic region, and we will find it more convenient to work in the $(a,b)$-plane. We will show the image of this parabola under our bijection $F_3$ is a line through the origin in the $(a,b)$-plane. We then show this line intersects at least one of the curves $F_o(a, b)=0$ and $F_e(a, b)=0$. Using these results, we can then prove Proposition~\ref{prop:intersectionsEP} later in the section.

We begin by establishing the images of the parabolas are lines under $F_3$:

\begin{lemma}[Transformation of quadratics from the $(\tau, \mu)$-plane into the $(a, b)$-plane]\label{lemma:parabola}
For each $c>1/4$, the parabola $\mu = -c\tau^2$ is mapped to the line $b=m(c)a$ in the $(a,b)$-plane, where the slope $m(c)$ satisfies
\[
0<m(c)=\frac{\sqrt{|1 - 4c|}}{\sqrt{1 + |1 - 4c|} + 1}<1,
\]
and $m(c)\to1$ as $c\to\infty$.
\end{lemma}

\begin{proof}
Since $\tau<0$ we have $\sqrt{\tau^2} = |\tau|$, and so by the bijection Lemma~\ref{lemma:Bij}(\ref{lemma:crit}), 
\begin{align*}
(a, b) &= F_3^{-1}(\tau, -c\tau^2)\\
&= \left(\frac{\sqrt{|\tau|}\sqrt{1 + \sqrt{1 + |1 - 4c|}}}{2}, \frac{\sqrt{|\tau|}\sqrt{-1 + \sqrt{1 + |1 - 4c|}}}{2}\right).
\end{align*} 
From this, we obtain a linear relationship $b=m(c) a$ with $m(c)$ as above.

\end{proof}

Next, we examine the intersections of lines with the eigenvalue curves.

\begin{lemma}[Intersections of eigenvalue curves $F_o(a, b)=0$ and $F_e(a, b)=0$ with a line]\label{lemma:intersections}
For each slope $0<m<1$, the line $b=ma$ intersects at least one of the curves $F_o(a, b)=0$ or $F_e(a, b)=0$ when $a\in(0, \pi/2)$ and $b\in(0, \pi/\sqrt{12})$. 

In particular, when $1/\sqrt{3}\leq m<1$, the line intersects $F_o(a, b)=0$ at least once. If $0< m\leq 1/\sqrt{3}$, the line intersects $F_e(a, b)=0$ at least once.
\end{lemma}

  \begin{figure}[H]
    \begin{center}
\includegraphics[scale=0.6]{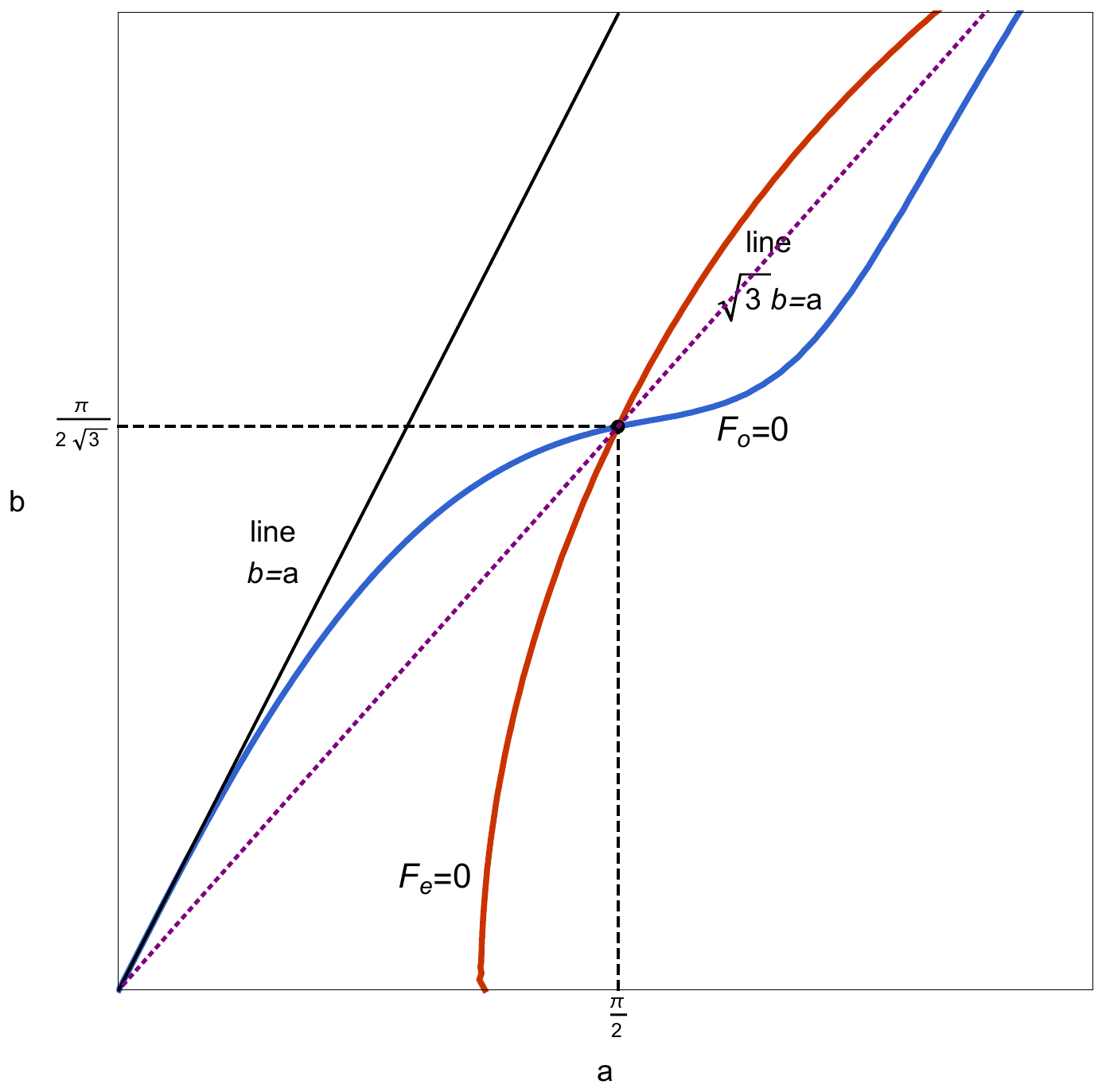}
\caption{\label{fig:(a, b)-plane} Two eigenvalue branches in the sub-parabolic region are transformed in the $(a, b)$-plane. The odd eigenvalue branch $F_o(a, b)=0$ (blue curve) and even eigenvalue branch $F_e(a, b)=0$ (red curve) intersect repeatedly along the dotted purple line $b=a/\sqrt{3}$. }
    \end{center}
    \end{figure}

\begin{proof}
Following Lemma~\ref{lemma:euoddeven}, we take $b(a)$ to be the solution of $F_o(a, b)=0$ for $a\in(0,\pi/2)$ and likewise take $a(b)$ to be the solution of $F_e(a, b)=0$ for $b\in(0,\pi/\sqrt{12})$. We shall use the Intermediate Value Theorem together with continuity of the solutions $b(a)$ and $a(b)$ (Lemma~\ref{lemma:smoothness}) to obtain existence of an intersection of the line and the curves $F_o(a, b)=0$ or $F_e(a, b)=0$.

We consider three cases on $m$. 

\begin{enumerate}
\item[(i)]\emph{For each $1/\sqrt{3}<m<1$, the line $b=ma$ intersects the curve $F_o(a, b)=0$.} 
We showed in Lemma~\ref{lemma:smoothness} that $b'(0^{+})=1$, so the line $b=ma$ lies below the graph of $b(a)$ near the origin. By direct computation, the line also lies above $b(a)$ at $a=\pi/2$ (see Figure~\ref{fig:(a, b)-plane}). By continuity, 
there exists some $a'\in(0, \pi/2)$ such that $b(a')=ma'$. 

\item[(ii)]\emph{For each $m= 1/\sqrt{3}$, the line $b=a/\sqrt{3}$ intersects the curves $F_e(a, b)=0$ and $F_o(a, b)=0$.} By Proposition~\ref{prop:critintersection}, the eigenvalue curves intersect at $(a, b)=(\pi/2, \pi/\sqrt{12})$, which lies on the desired line.  

\item[(iii)]\emph{For each $0< m< 1/\sqrt{3}$, the line $b=ma$ intersects the curve $F_e(a, b)=0$.} The line $a=b/m$ lies below the curve $a(b)$ at $b=0$ and lies above the curve at $b=\pi/\sqrt{12}$ (see Figure~\ref{fig:(a, b)-plane}). By continuity, 
there exists some $b'\in(0, \pi/\sqrt{12})$ such that $a(b')=b'/m$. 
\end{enumerate}

Finally, from Lemma~\ref{lemma:parabola}, we see $m_1=1/\sqrt{3}$; when $1<c<\infty$, we have $1/\sqrt{3}<m<1$; and when $1/4<c<1$, we have $0< m< 1/\sqrt{3}$.
\end{proof}

\begin{proof}[Proof of Proposition~\ref{prop:intersectionsEP}]
For every $0<c\leq 1/4$, the parabola $\mu=-c\tau^2$ lies on or above the critical parabola. In the super-parabolic region, by Theorem~\ref{thm:paramLHP}(\ref{thm:paramLHPEven}), there is a continuous eigenvalue curve 
that connects the point $(-\pi^2/4, 0)$ on the $\tau$-axis with the point $C(-2(a^\ast)^2, -(a^\ast)^4)$ on the critical parabola. This curve is in fact the first even eigenvalue curve in the super-parabolic region. At the upper point $(-\pi^2/4, 0)$, we have $-c\tau^2<\mu$. At the lower point $C(-2(a^\ast)^2, -(a^\ast)^4)$, we have $-c\tau^2>\mu$. Therefore, since $\mu$ and $\tau$ are continuous along the curve, there exists some point at which $-c\tau^2=\mu$.  

For the case of $c>1/4$, the parabola lies in the sub-parabolic region and we work in the $(a, b)$-plane. By Lemma~\ref{lemma:parabola}, we have that the parabola $\mu=-c\tau^2$ maps to the line $b=m(c)a$, with $1<c<\infty$ corresponding to $1/\sqrt{3}<m(c)<1$, with $1/4<c<1$ corresponding to $0<m(c)<1/\sqrt{3}$, and with $m(1)=1/\sqrt{3}$. By Lemma~\ref{lemma:intersections}, we showed that each line intersects at least one of the curves $F_o(a, b)=0$ and $F_e(a, b)=0$, which corresponds to eigenvalue curves in the sub-parabolic region by using the homeomorphism $F_3$. 
\end{proof}

\subsection{\bf Properties of the eigenvalue curves in the sub-parabolic region}\label{sec:critproperty}
We end with a result on the intersections of the eigenvalue curves with each other.
\begin{proposition}[Intersections of Eigenvalue Curves] \label{prop:critintersection}
The odd and even eigenvalue branches in the sub-parabolic region intersect infinitely often along the parabola $\mu=-\tau^2$.
\end{proposition}
\begin{proof}
The odd and even eigenvalue curves intersect when $a$ and $b$ both satisfy eigenvalue conditions \eqref{eqn:critevalCondOdd} and \eqref{eqn:critevalCondEven} simultaneously. By adding the two equations, we obtain
\[
2(3a^2-b^2)\frac{\sin(2a)}{2a} = 0.
\]
Hence the equation gives us (i) $3a^2=b^2$ or (ii) $\sin(2a) = 0$. Case (i) is not possible since $a > b$ by properties of $F_3$. When (ii) is satisfied, we have $a = m\pi/2$ for some $m\in \mathbb{N}$ (see Figure~\ref{fig:critEvalcondition}), and furthermore,
\[
0 = (3b^2 - a^2)\frac{\sinh(2b)}{2b}. 
\]
But since $\sinh(2b)/2b >0$ for all $b>0$, the above reduce to $3b^2=a^2$. The $(a,b)$ values at the points of intersection correspond to
\[
\tau = -\frac{4}{3}a^2,\quad \mu = -\frac{16}{9}a^4, \qquad\text{where $a=\frac{m}{2}\pi$, $m\in \mathbb{N}$}.
\]
These $(\tau,\mu)$ all lie on the parabola $\mu=-\tau^2$, as desired (see Figure~\ref{fig:critintersection}).
\end{proof}

  \begin{figure}[H]
    \begin{center}
\includegraphics[scale=0.6]{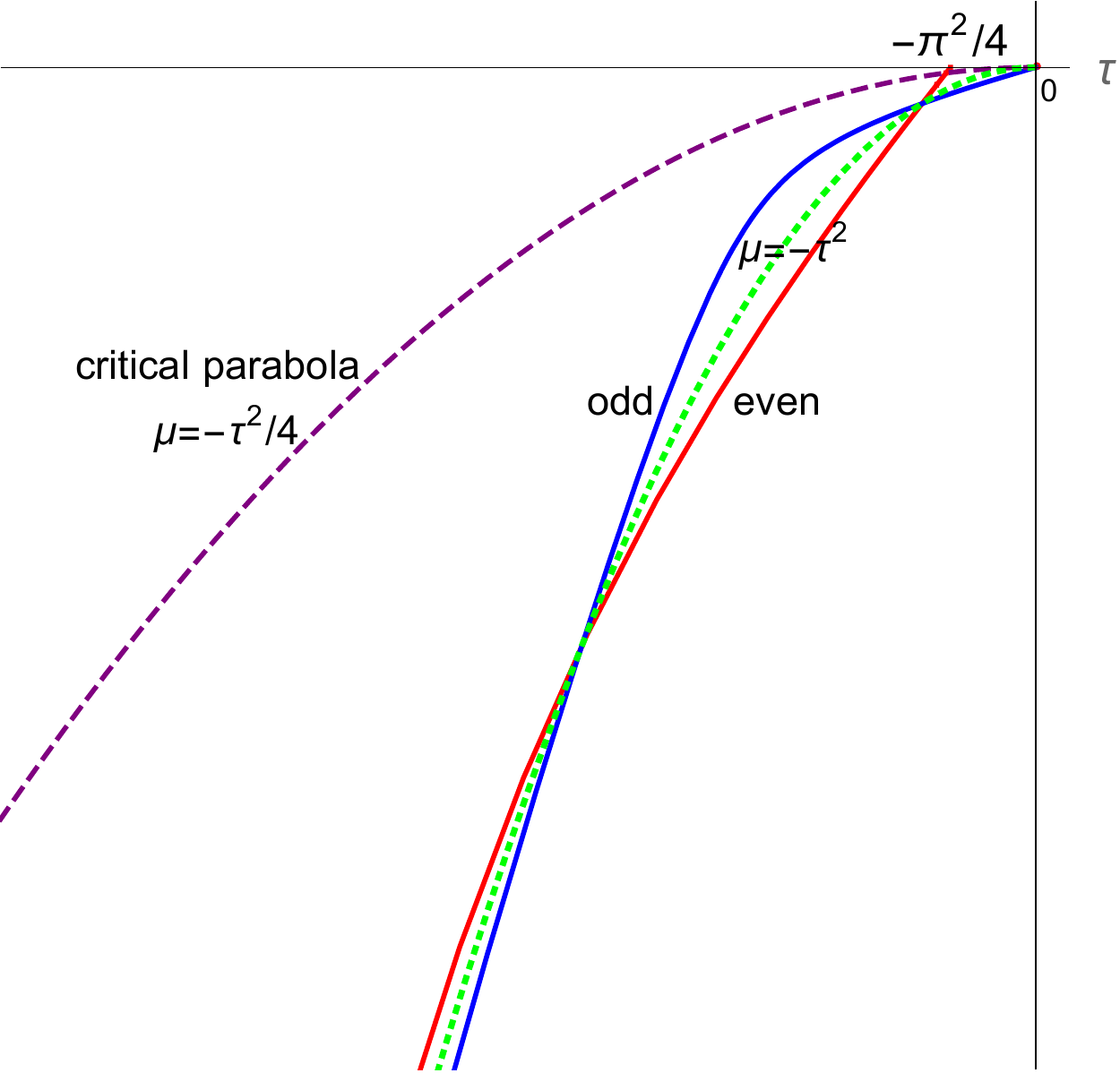}
\caption{\label{fig:critintersection} Sub-parabolic region: the odd and even eigenvalue curves intersect on the parabola $\mu = -\tau^2$ at the points $(\tau, \mu) = (-m^2\pi^2/3, -m^4\pi^4/9)$. See Proposition~\ref{prop:critintersection} in Section~\ref{sec:critproperty}.}
    \end{center}
    \end{figure}

\section*{Acknowledgments}
Chung's research was supported by the Department of Mathematics, University of Illinois.

\end{document}